\documentclass[pra,onecolumn,11pt,amsmath,amssymb,nofootinbib,bm,bbm]{revtex4-1}

\usepackage[english]{babel}  

\usepackage{natbib,amsmath, amsthm, bbm, tensor, mathtools, placeins, graphicx, wrapfig, subfigure, sidecap, float, verbatim}

\setcitestyle{sort&compress,numbers}
\usepackage[unicode=true, breaklinks=false, pdfborder={0 0 1}, backref=false, colorlinks=true, linkcolor=blue, citecolor=blue]{hyperref}

\usepackage{diagbox,multirow,scalerel}
\usepackage{makecell}

\newcommand{\Acal}{\mathcal{A}}
\newcommand{\Ecal}{\mathcal{E}}
\newcommand{\Fcal}{\mathfrak{E}}
\newcommand{\Hcal}{\mathcal{H}}
\newcommand{\Ical}{\mathcal{I}}
\newcommand{\Mcal}{\mathcal{M}}
\newcommand{\Tcal}{\mathcal{T}}
\newcommand{\Ucal}{\mathcal{U}}

\newcommand{\tr}{\mbox{tr}}
\newcommand{\an}{{a}}
\newcommand{\s}{{s}}
\newcommand{\e}{{e}}
\newcommand{\se}{{se}}
\usepackage{cancel}
\usepackage{color}
\definecolor{Blue}{rgb}{0,0,1}
\definecolor{Red}{rgb}{1,0,0}
\definecolor{Green}{rgb}{0.0, 0.5, 0.0}
\definecolor{Purp}{rgb}{0.87, 0.36, 0.51}
\definecolor{white}{rgb}{1,1,1}

\def\BraVert{\egroup\,\mid\,\bgroup}

\def\ketbra#1#2{\ket{#1\vphantom{#2}}\!\bra{#2\vphantom{#1}}}

\def\bra#1{\mathinner{\langle{#1}|}}
\def\bbra#1{\mathinner{\langle \! \langle{#1}|}}
\def\ket#1{\mathinner{|{#1}\rangle}}
\def\kket#1{\mathinner{|{#1}
\rangle \! \rangle}}
\def\braket#1{\mathinner{\langle{#1}\rangle}}

\newtheorem{lemma}{Lemma}

\begin{document}

\title{An introduction to operational quantum dynamics}

\author{Simon Milz}
\email{simon.milz@monash.edu}
\affiliation{School of Physics and Astronomy, Monash University, Clayton, Victoria 3800, Australia}

\author{Felix A. Pollock}
\email{felix.pollock@monash.edu}
\affiliation{School of Physics and Astronomy, Monash University, Clayton, Victoria 3800, Australia}

\author{Kavan Modi}
\email{kavan.modi@monash.edu}
\affiliation{School of Physics and Astronomy, Monash University, Clayton, Victoria 3800, Australia}

\begin{abstract}
 In the summer of 2016, physicists gathered in Toru{\'n}, Poland for the 48th annual \emph{Symposium on Mathematical Physics}. This Symposium was special; it celebrated the 40th anniversary of the discovery of the Gorini-Kossakowski-Sudarshan-Lindblad master equation, which is widely used in quantum physics and quantum chemistry. This article forms part of a \emph{Special Volume} of the journal \emph{Open Systems~\&~Information Dynamics} arising from that conference; and it aims to celebrate a related discovery -- also by Sudarshan -- that of \textit{Quantum Maps} (which had their 55th anniversary in the same year). Nowadays, much like the master equation, quantum maps are ubiquitous in physics and chemistry. Their importance in quantum information and related fields cannot be overstated. In this manuscript, we motivate quantum maps from a tomographic perspective, and derive their well-known representations. We then dive into the murky world beyond these maps, where recent research has yielded their generalisation to non-Markovian quantum processes.

\end{abstract}

\date{\today}
\maketitle
\tableofcontents

\pagebreak

Describing changes in a system's state is the principal goal of any mathematical theory of dynamics. In order to be physically relevant, this description must be faithful to what is observed in experiments. For quantum systems, a dynamical theory must quantify how measurement statistics of different observables can change from one moment to the next, even when the system in question may be interacting with its wider environment, which is typically large, uncontrollable and experimentally inaccessible.

While unitary evolution of vectors in Hilbert space (according to Schr{\"o}dinger's equation) is sufficient to describe the behaviour of a deterministically prepared closed quantum system, more is required when the system is open to its environment, or when there is classical uncertainty in its preparation. The complete statistical state of such a system (or, more properly, an ensemble of identical and independent preparations of the system) is encoded in its density operator $\rho$, which can be determined operationally in a quantum state tomography experiment. Namely, by combining the measurement statistics of a set of linearly independent observables. A reader who is unfamiliar with the concept of density operators or quantum state tomography can find more information in Ref.~\cite{Nielsen00a}.

In this Special Issue article, we concern ourselves with the dynamical description of open quantum systems, primarily in terms of mappings from density operators at one time or place to another, \textit{i.e.}, the quantum generalisation of classical stochastic maps. 
These mappings are \emph{superoperators} -- operators on an operator space -- and depending on context, are referred to as \textit{quantum maps}, \textit{quantum channels}, \textit{quantum operations}, \textit{dynamical maps}, and so on. In this article, we stick with the term \textit{quantum maps} throughout. 

Quantum maps are ubiquitous in the quantum sciences, particularly in quantum information theory. They are natural for describing quantum communication channels~\cite{wilde2013quantum}, crucial for quantum error correction~\cite{lidarQEC}, and form the basis for generalised quantum measurements~\cite{peresQT}. Yet their origins, motivation and applicability are not always transparent. Their discovery dates back to 1961, in the work of George Sudarshan and collaborators~\cite{SudarshanMatthewsRau61, SudarshanJordan61}. A decade later, Karl Kraus also discovered them~\cite{kraus_general_1971}, and quantum maps are perhaps most widely known through his 1984 book~\cite{kraus_states_1983}. Along the way, there have been many other players. For example, the works of Stinespring~\cite{stinespring1955} and Choi~\cite{choi72a, choi75} are crucial for understanding the structure of quantum maps. Stinespring's result predates that of Sudarshan, though both his and Choi's works are purely mathematical in nature and not concerned with quantum physics \textit{per se}. On the physics side, the works of Davies and Lewis~\cite{Davies1970}, Jamiolkowski~\cite{jamiolkowski_linear_1972}, Lindblad \cite{Lindblad1975}, and Accardi \textit{et al.}~\cite{accardi_quantum_1982}, to name but a few\footnote{A complete list of important contributions would constitute an entire article in itself, and we apologise to any who feel they have been unjustly omitted.}, have led to a deep understanding of quantum stochastic processes. Here, we put history aside, and describe quantum maps and their generalisations in a pedagogical manner. We present an operationally rooted and thorough introduction to the theory of open quantum dynamics.

The article has two main sections. In Section~\ref{sec::QuantumMaps}, we introduce quantum maps in the context of quantum process tomography -- that is, what can be inferred about the evolution of the density operator in experiment -- before exploring how they can be represented mathematically. Along the way, we take care to point out the relationships between different representations, and the physical motivation behind mathematical properties such as linearity and complete positivity. In Section~\ref{sec::beyond}, we discuss open quantum dynamics in situations where the formalism developed in the first Section is insufficient to successfully describe experimental observations. Namely, when the system is initially correlated with its environment and when joint statistics across multiple time steps is important. After demonstrating how a na{\"i}ve extension of the conventional theory fails to deliver useful conclusions, we outline a more general, operational framework, where evolution is described in terms of mappings from preparations to measurement outcomes.

\section{Quantum maps and their representations}
\label{sec::QuantumMaps}

A quantum map $\Ecal$ is a mapping from density operators to density operators: $\rho \mapsto \rho' = \Ecal[\rho]$. Here $\rho$ and $\rho'$ are operators on the `input' and `output' Hilbert spaces of the map, respectively\footnote{Strictly speaking, the mapping is between a preparation that yields $\rho$, and a measurement that interrogates $\rho'$.}. Formally, this can be written $\Ecal: \mathcal{B} (\Hcal_{d_\mathrm{in}}) \rightarrow \mathcal{B} (\Hcal_{d_\mathrm{out}})$, \textit{i.e.}, as a mapping from bounded operators on the input Hilbert space to bounded operators on the output Hilbert space. In fact, the map can be seen as a bounded operator on the space of bounded operators, $\Ecal\in \mathcal{B}(\mathcal{B}(\mathcal{H}_{d_\mathrm{in}}))$. Here $d = \mathrm{dim}(\Hcal_{d})$ denotes the dimension of Hilbert space $\Hcal_{d}$. Throughout this article, we work in the Schr{\"o}dinger picture with finite dimensional quantum systems (see Ref.~\cite{kretschmann_quantum_2005} for a description of quantum maps in the Heisenberg picture). In general, the input and output Hilbert space need not be the same, but for simplicity we will, for the most part, assume $\Hcal_{d_\mathrm{in}} \cong \Hcal_{d_\mathrm{out}}$ and omit the subscripts ``in'' and ``out'' from this point on.

To represent a deterministic physical process, the quantum map has to preserve the basic properties of the density operator, \textit{i.e.}, it has to preserve trace, Hermiticity, and positivity (as we detail more explicitly at the end of this section). Moreover, the action of the quantum map must be linear: 
\begin{gather}
\Ecal \left[\sum p_k \rho_k \right] = \sum p_k \Ecal[ \rho_k] = \sum p_k \rho'_k.
\end{gather}
It is worth noting that this requirement does \emph{not} follow from the fact that quantum mechanics is a linear theory, in the sense of quantum state vectors formed from linear superpositions of a basis set (in fact, $\Ecal$ is not generally linear in this sense). Instead, the linearity of the quantum map is analogous to the linearity of mixing in a statistical theory.

To better appreciate this, consider a quantum channel from Alice to Bob, where Alice prepares a system in either state $\rho_1$ or $\rho_2$; she then sends the system to Bob. Upon receiving the system Bob performs state tomography on the state Alice sent by measuring it. They do this many times. Suppose Alice sends $\rho_1$ on day-one and $\rho_2$ on day-two. From the measurement outcomes Bob will conclude that the received states are $\rho'_1 = \Ecal[\rho_1]$ on day-one and $\rho'_2 =\Ecal[\rho_2]$ on day-two. Now, suppose Alice sends the two states randomly with probabilities $p$ and $1-p$ respectively. Without knowing which state is being sent on which run, Bob would conclude that he receives state $\bar \rho' = \Ecal [\bar \rho]$, where $\bar \rho = p \rho_1 + (1-p) \rho_2$. That is, we can interpret Alice's preparation to be the average state. Now suppose that, at some later point, Alice reveals which state was sent in which run; Bob can now go back to his logbook and conclude that he received the state $\rho'_1 \, (\rho'_2)$ whenever Alice sent him $\rho_1 \, (\rho_2)$. Conversely, averaging over that data would amount to Bob receiving $\bar \rho '$. Thus we must have $\bar \rho' = p \rho'_1 + (1-p) \rho'_2$. This simple thought experiment demands that the action of quantum channels must be linear. However, note that, while we have used the language of quantum mechanics in this paragraph, there is nothing quantum about this experiment\footnote{The same argument would hold for a nonlinear map on the space of pure states. However, care has to be taken in differentiating between proper and improper mixtures~\cite{mixtures}.}. Linearity of mixing is a general concept that applies to all stochastic theories.

Now, using the fact that the quantum map $\Ecal$ is linear, we will derive several useful representations for it.

\subsection{Structure of linear maps}
\label{subsec::LinMaps}

\textit{Any} linear map $M$ on a (complex) vector space $V$ is unambiguously defined by its action on a (not necessarily orthogonal) basis $\left\{\hat{\mathbf{r}}_i\right\}_{i=1}^{d_V}$\footnote{Here, and throughout this article, the caret is used to indicate that the object is an element of some fixed (not necessarily normalised) basis set used for tomography.} of $V$, where $d_V$ is the dimension of $V$. That is to say, the input-output relation $M[\hat{\mathbf{r}}_i] = \hat{\mathbf{r}}_i'$ entirely defines $M$. For any basis $\left\{\hat{\mathbf{r}}_i\right\}$ of $V$, there exists a \emph{dual set} $\left\{\hat{\mathbf{d}}_i\right\}_{i=1}^{d_V}\subset V$ such that $(\hat{\mathbf{d}}_i,\hat{\mathbf{r}}_j) = \delta_{ij}$, where $(\cdot,\cdot)$ is the scalar product in $V$. With this, for any $\mathbf{v}\in V$, the action of $M$ can be written as 
\begin{gather}
\label{eqn::ActionLinear}
M[\mathbf{v}] = \sum_{i=1}^{d_V} \hat{\mathbf{r}}_i'\, (\hat{\mathbf{d}}_i,\mathbf{v})\, .
\end{gather}
This equation is correct by construction, as it maps every basis element $\hat{\mathbf{r}}_i$ to the correct output $\hat{\mathbf{r}}_i'$. In other words, it says that knowing the images under a map $M: V\rightarrow V$ for a basis of $V$ completely defines the action of the map.

Eq.~\eqref{eqn::ActionLinear} can be rewritten as
\begin{gather}
\label{eqn::OuterProd}
M[\mathbf{v}] = \sum_{i=1}^{d_V} \hat{\mathbf{r}}_i'\, (\hat{\mathbf{d}}_i,\mathbf{v}) \equiv \sum_{i=1}^{d_V} \left(\hat{\mathbf{r}}_i' \times \hat{\mathbf{d}}_i^{*}\right)[\mathbf{v}]\, ,
\end{gather}
where we have defined the \emph{outer product}
\begin{gather}
\label{eqn::ComponentsOuter}
\left(\hat{\mathbf{r}}_i' \times \hat{\mathbf{d}}_i^{*}\right)_{kl} =  (\hat{\mathbf{r}}_i')_k(\hat{\mathbf{d}}_i)^{*}_l\, .
\end{gather}
For an orthonormal basis $\left\{\hat{\mathbf{e}}_i\right\}$ of $V$, we have $(\mathbf{a})_i= (\mathbf{a},\hat{\mathbf{e}}_i)$, and $N_{ij} = (\hat{\mathbf{e}}_i, N[\hat{\mathbf{e}}_j])$ for any $\mathbf{a} \in V$ and any linear operator $N$ on $V$. Consequently, we obtain a matrix representation $\mathbf{M}$ of the map $M$:
\begin{gather}
\label{eqn::MatLambda}
(\mathbf{M})_{kl} = \sum_{i=1}^{d_V} \left(\hat{\mathbf{r}}_i' \times \hat{\mathbf{d}}_i^{*}\right)_{kl} 
\end{gather}
and the action of $M$ can be written in terms of the matrix $\mathbf{M}$:
\begin{gather}
(M[\mathbf{v}])_k = \sum_l (\mathbf{M})_{kl} \, v_l = \sum_{i=1}^{d_V} \left(\hat{\mathbf{r}}_i' \times \hat{\mathbf{d}}_i^{*}\right)_{kl} \, v_l\, ,
\end{gather}
where $\mathbf{v} = \sum_m v_m \hat{\mathbf{e}}_m$. Note that there is a distinction between $M$ and $\mathbf{M}$; the former is a map, while the latter is its representation as a matrix. This distinction is often not made when dealing with quantum maps, but here we will make it explicit.

\begin{figure}
\centering
\subfigure[]
{
\includegraphics[scale=1.4]{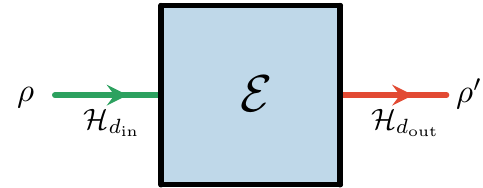}
}
\hspace{1.1cm}
\subfigure[]
{
\includegraphics[scale=0.8]{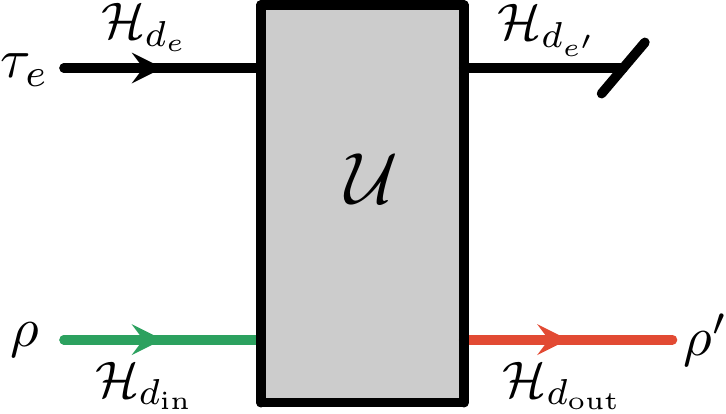}
}
\caption{
\textbf{(a)} \emph{Action of a linear map} $\Ecal:\mathcal{B}(\Hcal_{d_\mathrm{in}})\rightarrow \mathcal{B}(\Hcal_{d_\mathrm{out}})$. The action of the map $\Ecal$ on a random $\rho$ is entirely defined by its 
input-output relation $\{\hat{\rho}_i\} \rightarrow \{\hat{\rho}_i'\}$. \textbf{(b)} \emph{Stinespring dilation}. Any completely positive trace preserving map can be represented (non-uniquely) as the contraction 
of a unitary dynamics $\mathcal{U}$ in a larger space, \textit{i.e.} $\Ecal[\rho] = \tr_{e}\left(\mathcal{U}[\rho\otimes \tau_e]\right)$, where $\mathcal{U}[X] = U X \, U^{\dagger}$ and 
$U$ is a unitary matrix (see Sec.~\ref{subsec::CP}).}
\label{fig::InputOutput}
\end{figure}

\subsubsection{Tomographic representation}

A quantum map $\Ecal$ is a linear map on the vector space $\mathcal{B}(\Hcal_d)$. Since $\mathcal{B}(\Hcal_d)$ is isomorphic to the vector space of $d \times d$ matrices (where $d$ is the dimension of $\Hcal_d$), we can make use of the natural inner product on the latter space, the Hilbert-Schmidt inner product $(\rho,\eta) = \tr(\rho^{\dagger} \, \eta)$, to define an inner product on the space of density operators. Consequently, we can express the action of $\Ecal$ in a way equivalent to Eq.~\eqref{eqn::ActionLinear}; different generalisations of the outer product defined in Eq.~\eqref{eqn::OuterProd} will then lead to different representations of $\Ecal$ (see Sec.~\ref{subsec::Forms}).

To proceed, we need a basis set of the input space. There always exists a set of operators that constitutes a (generally non-orthogonal) basis of $\mathcal{B}(\Hcal_d)$. For example, the set of elementary matrices form such a basis, as do Pauli and Gell-Mann matrices. Both of these sets are orthonormal with respect to Hilbert-Schmidt inner product, but neither of them consists of physical density operators. However, as explained above, the map $\Ecal$ is unambiguously defined by its input-output relation $\Ecal[\hat{\rho}_i] = \hat{\rho}_i'$. Thus, we can use density operators for the basis set: $\left\{\hat{\rho}_i\right\}_{i=1}^{d^2} \subset \mathcal{B}(\Hcal_d)$. For example, for a two-level quantum system we can use the following density operators
\begin{gather}\label{eq:basisstates}
\hat{\rho}_1 = \frac{1}{2}\begin{pmatrix*}[r] 1 & 1 \\ 1 & 1   \end{pmatrix*}, \,
\hat{\rho}_2 = \frac{1}{2}\begin{pmatrix*}[r] 1 & -i \\ i & 1  \end{pmatrix*}, \,
\hat{\rho}_3 = \begin{pmatrix*}[r]1 & 0 \\ 0 & 0  \end{pmatrix*}, \,
\hat{\rho}_4 = \frac{1}{2}\begin{pmatrix*}[r] 1 & -1 \\ -1 & 1  \end{pmatrix*}.
\end{gather}
These matrices are linearly independent and form a basis, but clearly, they are not orthonormal. However, for any choice of basis, there exists a set of \emph{dual matrices} $\left\{\hat{D}_i\right\}_{i=1}^{d^2}$~\cite{modi_positivity_2012} such that $\tr(\hat{D}_i^{\dagger} \hat{\rho}_j) = \delta_{ij}$ (see App.~\ref{subsec::DualMat} for proof). Consequently, in analogy to Eq.~\eqref{eqn::ActionLinear}, the action of $\Ecal$ on $\rho$ can be written as 
\begin{gather}
\label{eqn::ActionMap}
\Ecal[\rho] = \sum_{i=1}^{d^2} \hat{\rho}_i' \, \tr(\hat{D}_i^{\dagger}\rho)\, ,
\end{gather}
which means that determining the output states for a basis of input states entirely defines the action of the map $\Ecal$.

The dual matrices for the states in Eq.~\eqref{eq:basisstates} are
\begin{gather}
\label{eqn::DualsEx}
\hat{D}_1 = \frac{1}{2}\begin{pmatrix*}[c] 0 & 1+i \\ 1-i & 2   \end{pmatrix*}, \,
\hat{D}_2 = \frac{1}{2}\begin{pmatrix*}[r] 0 & -i \\ i & 0   \end{pmatrix*}, \,
\hat{D}_3 = \begin{pmatrix*}[r] 1 & 0 \\ 0 & -1   \end{pmatrix*}, \,
\hat{D}_4 = \frac{1}{2}\begin{pmatrix*}[c] 0 & -1+i \\ -1-i & 2   \end{pmatrix*}.
\end{gather}
Clearly these dual matrices are not positive. In fact, if both the outputs $\hat{\rho}_i'$ and the duals are positive, then $\Ecal$ is necessarily an entanglement breaking channel~\cite{horodecki_entanglement_2003,holevo_1998} (the converse also holds). In general, neither set of matrices in Eq.~\eqref{eqn::ActionMap}, $\{\hat{\rho}_i'\}$ and $\{\hat{D}_i\}$, have to be positive, and it can sometimes even be advantageous to choose non-positive matrices $\hat{\rho}_i'$ and $\hat{D}_i$ for the representation of $\Ecal$. 

However, for a proper quantum map, we can choose $\{\hat{\rho}_i'\}$ to be states, and $\{\hat{D}_i\}$ to be the dual set corresponding to a set of basis states $\{\hat{\rho}_i\}$. Then Eq.~\eqref{eqn::ActionMap} captures  precisely the idea of \emph{quantum process tomography}~\cite{JModOpt.44.2455, PhysRevLett.78.390}, where  the dynamics of a quantum system is experimentally reconstructed by relating a basis of input states to their corresponding outputs. The action of the map $\Ecal$ on any state $\rho$ is then simply determined from the linearity of the map. From here on, we will -- for obvious reasons -- refer to this representation as the \emph{input/output} or \emph{tomographic} representation of $\Ecal$.

\subsubsection{Operator-sum representation}

Based on Eq.~\eqref{eqn::ActionMap}, the action of $\Ecal$ can be rewritten in a form that is used more widely in the literature. Both $\hat{\rho}_i'$ and $\hat{D}_i$ can be expressed in terms of their left- and right-singular vectors, \textit{i.e.},
\begin{gather}
\hat{\rho}_i' = \sum_{\alpha} \ketbra{s^i_\beta}{t^i_\beta}\, , \quad \text{and} \quad \hat{D}_i = \sum_\mu \ketbra{u^i_\mu}{v^i_\mu}\, ,
\end{gather}
where $\{\ket{s_\alpha^i}\}, \{\ket{t_\alpha^i}\}$ and $\{\ket{u_\mu^i}\}, \{\ket{v_\mu^i}\}$ are the respective unnormalised left- and right-singular vectors of $\hat{\rho}_i'$ and $\hat{D}_i$. With this decomposition, the action of $\Ecal$ reads
\begin{align}
\Ecal[\rho] &= \sum_i\hat{\rho}_i'\, \tr(\hat{D}_i^{\dagger}\rho) = \sum_i\sum_{\beta,\mu} \ketbra{s^i_\beta}{t^i_\beta} \tr\left(\ketbra{v^i_\mu}{u^i_\mu} \rho\right) \\
&= \sum_{\beta,\mu} \sum_i \left(\ketbra{s^i_\beta}{u^i_\mu}\right) \rho \left(\ketbra{v^i_\mu}{t^i_\beta}\right)\, .
\end{align}
Compressing the indeces $\{i,\beta,\mu\}$ into one common index yields the \emph{operator sum representation} of $\Ecal$:
\begin{gather}
\label{eqn::OpSum}
\Ecal[\rho] = \sum_{\beta,\mu} \sum_i \left(\ketbra{s^i_\beta}{u^i_\mu}\right)\rho \left(\ketbra{v^i_\mu}{t^i_\beta}\right) \equiv \sum_\alpha L_\alpha \rho \, R_\alpha^{\dagger}\, ,
\end{gather}
where $L_\alpha$ and $R_\alpha$ have the same shape, but are not necessarily square (if the input and output space are not of the same size). In exactly the same vein, the tomographic representation of a map can be recovered from its operator sum representation via a singular value decomposition.

\textbf{Unitary freedom.} We have shown that any linear map can be expressed in the operator sum representation, but the set of matrices $\{L_\alpha,R_\alpha\}$ in  Eq.~\eqref{eqn::OpSum} is not unique. Any set $\{L'_\mu,R'_\mu\}$ of matrices that is connected to $\{L_\alpha,R_\alpha\}$ by an isometry, \textit{i.e.}, $L'_\mu = \sum_{\alpha} (U)_{\mu\alpha}L_{\alpha}$ and $R'_\mu = \sum_{\alpha'} (U)_{\mu\alpha'}R_{\alpha'}$, where $U^{\dagger}U = \openone$, gives rise to the same linear map:
\begin{gather}
\sum_{\mu}L'_{\mu}\rho R_{\mu}^{\prime\,\dagger} = \sum_{\alpha\alpha'}\sum_{\mu} (U)_{\mu\alpha}L_{\alpha} \rho  R_{\alpha'}^{\dagger}(U)_{\mu\alpha'}^{*} = \sum_{\alpha\alpha'} (U^{\dagger}U)_{\alpha'\alpha} L_{\alpha'} \rho  R_{\alpha'}^{\dagger} = \sum_{\alpha} L_\alpha \rho R_\alpha^{\dagger}.
\end{gather}

Both of the representations we have presented so far consist of sets of operator pairs -- $\{\hat{\rho}_i',\hat{D}_i'\}$ for the tomographic representation and $\{L_\alpha,R_\alpha\}$ for the operator sum representation. These will be explored further later in this section. Next, however, we will present two matrix representations for the map.

\subsection{Matrix representations}
\label{subsec::Forms}

Since $\mathcal{B}(\mathcal{H}_d)$ is itself a vector space, it should be possible to represent $\Ecal$ -- a linear map on that space -- as a matrix. Indeed, two such representations were first discovered back in 1961~\cite{SudarshanMatthewsRau61}. To derive these representations, we note that there are (at least) two different ways to generalise the outer product Eq.~\eqref{eqn::ComponentsOuter}, and hence two different ways to obtain representations of $\Ecal$ in terms of outputs and dual matrices.

\subsubsection{Sudarshan's A form}

In clear analogy to Eq.~\eqref{eqn::MatLambda}, one possible matrix representation of $\Ecal$ (in an orthonormal basis of $\Hcal_d\otimes\Hcal_d$) is given by 
\begin{gather}
\label{eqn::EcalA}
\Ecal_A = \sum_{i=1}^{d^2} \hat{\rho}_i' \times \hat{D}_i^{*}\, ,\qquad \text{with} \quad \left(\hat{\rho}_i' \times \hat{D}_i^{*} \right)_{rs;r's'} = (\hat{\rho}_i')_{rs} (\hat{D}_i)^*_{r's'}\, .
\end{gather}
In Dirac notation, this means that we have generalised the outer product defined in \eqref{eqn::ComponentsOuter} as
\begin{gather}
\label{eqn::OuterMatrix}
\ketbra{r}{s}\times \ketbra{r'}{s'} \equiv \ketbra{rs}{r's'}.
\end{gather}
The action of $\Ecal$ can be simply written as 
\begin{gather}
\label{eqn::ActionA}
(\Ecal[\rho])_{rs} = \sum_{r's'}^d(\Ecal_A)_{rs;r's'}(\rho)_{r's'}.
\end{gather}
This is what Sudarshan \textit{et al.} called the A form of the dynamical map~\cite{SudarshanMatthewsRau61}. They observed that the matrix $\Ecal_A$ is not Hermitian even if $\Ecal$ is Hermiticity preserving. Indeed, this matrix is -- quite naturally -- not even square if the input dimensions are different from the output dimensions. 

Mathematically, the outer product `flips' the bra (ket) $\bra{s}$ $(\ket{r'})$ into the ket (bra) $\ket{s}$ $(\bra{r'})$. By \emph{vectorizing} $\rho$ and $\Ecal[\rho]$, we can write Eq.~\eqref{eqn::ActionA} in a more compact way:
\begin{gather}
\label{eqn::VecAction}
\kket{\Ecal[\rho]} = \Ecal_A \kket{\rho}, 
\quad \mbox{where} \quad 
\kket{\rho} = \sum_{rs}(\rho)_{rs} \ket{rs}
\quad \mbox{for} \quad
\rho = \sum_{rs} (\rho)_{rs} \ketbra{r}{s}.
\end{gather}
For the details of \emph{vectorisation} of matrices see, \textit{e.g.}, Refs.~\cite{dariano_orthogonality_2000,gilchrist_vectorization_2009}. Because the action of $\Ecal_A$ onto $\kket{\rho}$ is simply a multiplication of a vector by a matrix, this representation is often favourable for numerical studies.

\begin{figure}
\centering
\includegraphics[scale=1.]{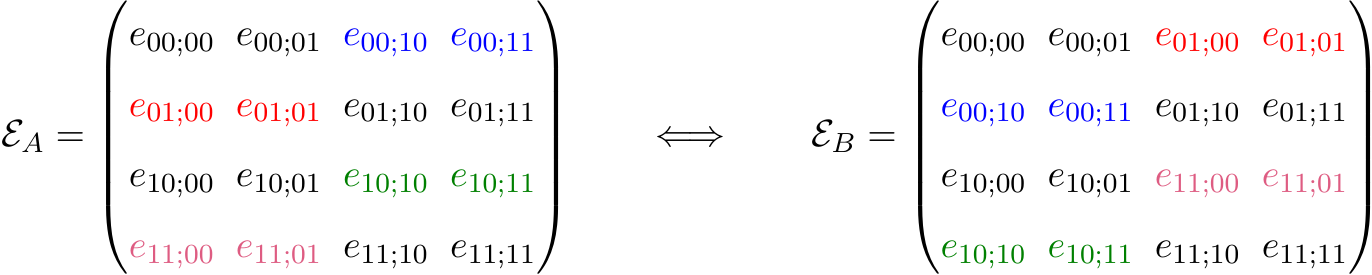}
\caption{\textit{Converting between A and B form.} In a given orthonormal basis, the A and B form of a map $\Ecal$ are related by a simple reshuffling of the matrix elements. For a better orientation, the matrix elements that change position are depicted in colour.}
\label{fig::Reshuffling}
\end{figure}

\subsubsection{Sudarshan's B form}

Next, we consider what Sudarshan \textit{et al.} called the B form of the dynamical map~\cite{SudarshanMatthewsRau61}. Instead of the outer product in Eq.~\eqref{eqn::EcalA}, let us consider a tensor product: 
\begin{gather}
\label{eqn::EcalB}
\Ecal_B = \sum_{i=1}^{d^2} \hat{\rho}_i' \otimes \hat{D}_i^{*},\quad \text{with} \quad \left(\hat{\rho}_i' \otimes \hat{D}_i^{*} \right)_{rr';ss'} = (\hat{\rho}_i')_{rs} (\hat{D}_i)^*_{s'r'}\, ;
\end{gather}
that is, the product in Eq.~\eqref{eqn::ComponentsOuter} is generalised to $\ketbra{r}{s}\otimes\ketbra{r'}{s'} = \ketbra{rr'}{ss'}$. The action of $\Ecal$ can be written as
\begin{gather}
\label{eqn::ActionB}
(\Ecal[\rho])_{rs} = \sum_{r's'}^d(\Ecal_B)_{rr';ss'}(\rho)_{r's'}.
\end{gather}

While the A form is closer in spirit to the general considerations about linear maps on vector spaces, the B form possesses nicer mathematical properties (see below), and from the point of view of quantum mechanics, the tensor product $\otimes$ seems `more natural' than the outer product $\times$. Comparing the matrices $\Ecal_A$ and $\Ecal_B$, it can be seen, from the relation between the outer product and the tensor product, that they coincide up to reshuffling~\cite{SudarshanMatthewsRau61, zyczkowski_duality_2004, bengtsson_geometry_2007}. In Fig.~\ref{fig::Reshuffling}, we show how to go between the two forms for a map acting on a two-level system (qubit). However, unlike $\Ecal_A$, the matrix $\Ecal_B$ is Hermitian iff $\Ecal_B$ is Hermiticity preserving. A quantum map is trace preserving iff $\tr_\mathrm{out}(\Ecal_B) = \openone_\mathrm{in}$, where $\tr_{\mathrm{out}}$ denotes the trace over the output Hilbert space of the map $\Ecal$ [\textit{i.e.}, the trace over the unprimed indices in Eq.~\eqref{eqn::ActionB}] and $\openone_\mathrm{in}$ is the identity matrix on the input space. We will prove these properties in the following subsections.

\subsection{Choi-Jamio{\l}kowski isomorphism}\label{sec::CJI}

Consider the action of a quantum map on one part of an (unnormalised) maximally entangled state $\ket{I} = \sum_{k=1}^{d_{\mathrm{in}}} \ket{kk}$:
\begin{gather}
\label{eqn::Choi}
\Upsilon_{\Ecal} = \Ecal \otimes \Ical \left[\ketbra{I}{I}\right] = \sum_{k,l=1}^{d_{\mathrm{in}}} \Ecal\left[\ketbra{k}{l}\right]\otimes \ketbra{k}{l}\, ,
\end{gather}
where $\{{\ket{k}}\}$ is an orthonormal basis of $\Hcal_{d_\mathrm{in}}$ and $\Ical$ is the identity operator on $\mathcal{B}(\Hcal_d)$. The resultant matrix $\Upsilon_{\Ecal}$ can be shown to be, element-by-element, identical to the quantum map $\Ecal$. In principal, any vector $\ket{I}$ with full Schmidt rank could be used for this isomorphism~\cite{dariano_imprinting_2003}. In the form of~\eqref{eqn::Choi} it is known as the \emph{Choi-Jamio{\l}kowski isomorphism} (CJI)~\cite{jamiolkowski_linear_1972, choi75}, an isomorphism between linear maps, $\Ecal: \mathcal{B} (\Hcal_{d_\mathrm{in}}) \rightarrow \mathcal{B} (\Hcal_{d_\mathrm{out}} )$, and matrices $\Upsilon_{\Ecal} \in \mathcal{B}(\Hcal_{d_\mathrm{out}})\otimes \mathcal{B}(\Hcal_{d_\mathrm{in}})$. In order to keep better track of the input and output spaces of the map $\Ecal$, we explicitly distinguish between the spaces $\Hcal_{d_\mathrm{in}}$ and $\Hcal_{d_\mathrm{out}}$ in this subsection.

Usually, $\Upsilon_{\Ecal}$ is called the \emph{Choi matrix} or \emph{Choi state} of the map $\Ecal$ (we will refer to it as the latter when it is a valid quantum state, up to normalisation)\footnote{By means of the CJI, the density matrix $\rho$ itself can also be considered the Choi state of a map $\Ecal: \mathbb{C} \rightarrow \mathcal{B}(\Hcal)$~\cite{chiribella_theoretical_2009}.}. Given $\Upsilon_{\Ecal}$, the action of $\Ecal$ can be written as 
\begin{gather}
\label{eqn::ChoiAction}
\Ecal[\rho] = \tr_{\mathrm{in}}\left[\left(\openone_{\mathrm{out}} \otimes \rho^{\mathrm{T}}\right)\Upsilon_\Ecal\right] \, ,
\end{gather}
where $\openone_{\mathrm{out}}$ is the identity matrix on $\Hcal_{d_\mathrm{out}}$ and $\tr_{\mathrm{in}}$ denotes the partial trace over the input Hilbert space $\Hcal_{d_\mathrm{in}}$. Eq.~\eqref{eqn::ChoiAction} can be shown by insertion:
\begin{align}
\tr_{\mathrm{in}}\left[\left(\openone_{\mathrm{out}} \otimes \rho^{\mathrm{T}}\right) \Upsilon_\Ecal\right] &=  \sum_{k,l}^{d_\mathrm{in}} \tr_{\mathrm{in}}\left[\left(\openone_{\mathrm{out}} \otimes \rho^{\mathrm{T}}\right) \left(\Ecal[\ketbra{k}{l}]\otimes \ketbra{k}{l} \right)\right] \\
&= \sum_{k,l,m}^{d_\mathrm{in}} \Ecal[\ketbra{k}{l}] \!\bra{m}\rho^{\mathrm{T}}\ket{k}\!\braket{l|m} = \sum_{k,l}^{d_\mathrm{in}} \rho_{kl} \Ecal[\ketbra{k}{l}] = \Ecal[\rho]\, .
\end{align}

The CJI is by no means restricted to quantum maps; \textit{any} linear map $\Ecal$ can be mapped to a Choi matrix $\Upsilon_\Ecal$ via the CJI. For instance, one can imprint a classical stochastic process onto a state by inputting one part of a maximally classically correlated state into the process. For quantum maps, however, the Choi matrix has particularly nice properties. Complete positivity of $\Ecal$ is equivalent to $\Upsilon_\Ecal \geq 0$ (see Sec.~\ref{Sec::OpSum} for a proof), and it is straightforward to deduce from Eq.~\eqref{eqn::ChoiAction} that $\Ecal$ is trace preserving iff $\tr_{\mathrm{out}}(\Upsilon_\Ecal) = \openone_\mathrm{in}$ (see Sec.~\ref{subsec::Quantum}).

Besides these appealing mathematical properties, the CJI is also of experimental importance. Given that a (normalised) maximally entangled state can in principal be created in practice, the CJI enables another way of reconstructing the map $\Ecal$, by letting it act on one half of a maximally entangled state and tomographically determining the resulting state. While this so-called \emph{ancilla-assisted process tomography}~\cite{PhysRevLett.90.193601,PhysRevLett.86.4195} requires the same number of measurements as the input-output procedure, it can be, depending on the experimental situation, easier to implement in the laboratory.

The mathematical properties of $\Upsilon_\Ecal$ are reminiscent of the properties of the B form. However, at first sight, it is not clear how the Choi matrix $\Upsilon_\Ecal$ is related to the different matrix representations of $\Ecal$ in terms of the dual matrices and outputs presented in Sec.~\ref{subsec::LinMaps}. The relation can be made manifest by using the fact that the set $\{\hat{\rho}_i\}_{i=1}^{d_\mathrm{in}^2}$ forms a basis of $\mathcal{B}({\Hcal_{d_\mathrm{in}}})$. With this, we can write $\ket{k}\bra{l} = \sum_{i=1}^{d_\mathrm{in}^2} \alpha_{i}^{(kl)} \hat{\rho}_i$, where $\alpha_{i}^{(kl)} \in \mathbb{C}$ is given by $\alpha_{i}^{(kl)} = \tr(\hat{D}_i^{\dagger} \ket{k}\bra{l})$. Consequently, we obtain
\begin{align}
\label{eqn::Choi_BForm}
\Upsilon_{\Ecal} &= \sum_{k,l} \Ecal[\ketbra{k}{l}]\otimes \ketbra{k}{l} = \sum_{k,l,i} \alpha_{i}^{(kl)}\,\Ecal[\hat{\rho}_i]\otimes  \ketbra{k}{l} \\
\label{eqn::Choi_BForm2}
&= \sum_i \Ecal[\hat{\rho}_i]\otimes  \sum_{k,l} \tr\left(\hat{D}_i^{\dagger} \ketbra{k}{l}\right) \ketbra{k}{l} = \sum_i \Ecal[\hat{\rho}_i] \otimes \hat{D}_i^*\, ,
\end{align}
where, in the last step, we have used the fact that $\{\ket{k}\bra{l}\}_{i,j=1}^{d_\mathrm{in}}$ also forms a basis of $\mathcal{B}(\Hcal_{d_\mathrm{in}})$. By comparison with Eq.~\eqref{eqn::EcalB}, we see that the Choi matrix of $\Ecal$ is exactly equal to the B form of $\Ecal$, \textit{i.e.}, $\Upsilon_\Ecal = \Ecal_B$, and henceforth, we will use the terms Choi matrix and B form interchangeably. 

\subsection{Operator sum representation revisited}
\label{Sec::OpSum}

As mentioned in Sec.~\ref{subsec::LinMaps}, any linear map can be written in terms of an operator sum representation. We proved this statement using the input-output action of the linear map, given in Eq.~\eqref{eqn::ActionLinear}. We now provide an alternative proof employing $\Upsilon_\Ecal$. The Choi matrix $\Upsilon_\Ecal$ can be written in terms of its unnormalised left- and right-singular vectors, \textit{i.e.} $\Upsilon_\Ecal = \sum_{\alpha=1}^{D} \ket{w_\xi}\bra{y_\xi}$, where $D= d_\mathrm{out} d_\mathrm{in}$. We have
\begin{align}
\Ecal[\rho] &= \sum_{\alpha=1}^D \tr_{\mathrm{in}}\left[\left(\openone_{\mathrm{out}} \otimes \rho^{\mathrm{T}}\right)\ketbra{w_\alpha}{y_\alpha}\right] = \sum_{\alpha=1}^D \sum_{k,l=1}^{d_\mathrm{in}}  \braket{l|w_\alpha}\!\bra{k}\rho^{\mathrm{T}} \ket{l}\! \braket{y_\alpha|k} \\
&= \sum_{\alpha=1}^D \left(\sum_l^{d_\mathrm{in}} \braket{l|w_\alpha}\! \bra{l}\right) \rho \left(\sum_{k=1}^{d_\mathrm{in}} \ket{k}\!\braket{y_\alpha|k}\right) \equiv \sum_{\alpha = 1}^D L_\alpha \rho \, R_\alpha^{\dagger}\,, 
\end{align}
which is the operator sum representation of $\Ecal$ already encountered in Sec.~\ref{subsec::LinMaps}.

Given the operator sum representation of a linear map $\Ecal$, it is possible to find another way of writing its A and B form. The B form $\Ecal_B$ is obtained via 
\begin{gather}\label{eqn::Kraus_B}
\Ecal_B = \Upsilon_\Ecal = \sum_\alpha \sum_{i,j = 1}^{d_\mathrm{in}} L_\alpha \ketbra{i}{j} R_\alpha^{\dagger} \otimes \ketbra{i}{j} = \sum_\alpha L_\alpha \times R_\alpha^*\, .
\end{gather}
Correspondingly, the A form of $\Ecal$ can be written as~\cite{usha_devi_open-system_2011}
\begin{gather}
\label{eqn::Kraus_A}
\Ecal_A = \sum_\alpha L_{\alpha} \otimes R_\alpha^{*}\, . 
\end{gather}
Indeed, substituting Eq.~\eqref{eqn::Kraus_A} into Eq.~\eqref{eqn::ActionA}, we obtain \begin{gather}
(\Ecal[\rho])_{rs} = \sum_{r's'}(\Ecal_A)_{rs;r's'}\rho_{r's'} = \sum_\alpha \sum_{r's'} (L_\alpha)_{rr'} \rho_{r's'} (R_\alpha^{*})_{ss'} = \left(\sum_\alpha L_\alpha \rho  R_\alpha^{\dagger}\right)_{rs}\, .
\end{gather}

The operators $\{ L_\alpha, \, R_\alpha \}$ are operationally different from $\{ \hat{\rho}_i', \hat{D}_i\}$ in Eq.~\eqref{eqn::EcalB}. A quantum map in the form $\Ecal_A$ ($\Ecal_B$) is obtained by tensor (outer) product of the former, and outer (tensor) product of the latter. Therefore, in clear analogy to the corresponding statement for the B form and the operator sum representation, we can recover the tomographic representation of $\Ecal$ via a singular value decomposition of $\Ecal_A$.

\subsection{Properties of quantum maps}
\label{subsec::Quantum}

The four representations derived above (input-output, operator sum, A form and B form) are valid for any linear map on a finite-dimensional complex operator space. However, not every such map represents the dynamics of a physical system. In order to do so, it must ensure that the statistical character of quantum states is preserved. Here, we lay out the mathematical constraints imposed on quantum maps by this requirement, and explore the corresponding implications for different representations.

\subsubsection{Trace preservation}
\label{subsubsec::Tr}

Since the trace of the density operator represents its normalisation, a deterministic quantum map should be trace preserving (more general, trace non-increasing maps do not have this property, and represent probabilistic quantum processes). This requirement can be stated as 
\begin{gather}
\tr(\rho') = \tr\left( \Ecal[\rho] \right) = \sum_i \tr[\hat{\rho}_i'] \, \tr[\hat{D}_i^\dag \rho] \qquad \forall\rho.
\end{gather}
Since $\tr[\hat{D}_i^\dag \hat{\rho}_i] = 1$, by linearity, the trace-preservation condition holds iff $\tr(\hat{\rho}_i') = \tr(\hat{\rho}_i)$. That is, the map $\Ecal$ is trace preserving iff it is trace preserving on a basis of inputs. Equivalently, a map $\Ecal$ is trace preserving iff $\sum_i (\tr\hat{\rho}_i')\hat{D}_i^{*} = \openone$.

Trace preservation can also be stated in a succinct way in terms of the operator sum representation. We have
\begin{gather}
\tr \left(\Ecal[\rho] \right) = \tr\left(\sum_\alpha L_\alpha \rho R_\alpha^{\dagger}\right) = \tr \left(\sum_\alpha R_\alpha^{\dagger} L_\alpha \rho \right)\, ,
\end{gather}
and hence $\Ecal$ is trace preserving for all $\rho$ iff $\sum_\alpha R_\alpha^{\dagger} L_\alpha= \openone$.

In a similar way, we can express trace preservation of $\Ecal$ in terms of the B form. If $\Ecal$ is trace preserving, we have $\tr(\Ecal[\rho]) = \rho$ for all $\rho$. In terms of $\Ecal_B$, this means
\begin{gather}
\label{eqn::TrChoi}
\tr(\Ecal[\rho]) = \tr\left[\left(\openone_{\mathrm{out}} \otimes \rho^{\mathrm{T}}\right) \Ecal_B\right] = \tr\left[\rho^{\mathrm{T}}\tr_\mathrm{out}(\Ecal_B)\right] = \tr(\rho)\, ,
\end{gather}
which is true iff $\tr_\mathrm{out}(\Ecal_B) = \openone$.

\subsubsection{Hermiticity preservation}
\label{subsubsec::Herm}

Given a valid quantum state as input, a physical quantum map should produce a valid quantum state as output; hence, it should preserve the Hermiticity of the density operator.
A map with this property satisfies $\Ecal[\rho] = \left(\Ecal[\rho]\right)^\dag$ for all $\rho = \rho^{\dagger}$. In terms of output matrices and duals, this condition reads
\begin{gather}
(\rho')^{\dagger} = \sum_i (\hat{\rho}_i')^{\dagger} \tr\left(\hat{D}_i \rho\right) = \sum_i \hat{\rho}_i' \tr\left(\hat{D}_i^{\dagger} \rho\right) = \rho'\, , \qquad \forall \ \rho = \rho^{\dagger}\, .
\end{gather}

If $\Ecal$ is Hermiticity preserving, then its B form (Choi matrix) $\Ecal_B$ is Hermitian. This follows from the fact that Hermiticity preservation of $\Ecal$ implies Hermiticity preservation of $\Ecal \otimes \Ical_\an$, where $\Ical_\an$ is the identity map on an arbitrary ancilla. Consequently, the decomposition of $\Ecal_B$ in terms of its left- and right-singular vectors becomes an eigendecomposition, \textit{i.e.}, $\Ecal_B = \displaystyle \sum_{\alpha = 1}^D\lambda_\alpha \ketbra{\alpha}{\alpha}$, where all the eigenvalues $\lambda_\alpha \in \mathbb{R}$ and we have $\braket{\alpha|\alpha'} = \delta_{\alpha\alpha'}$. Hence, the action of $\Ecal$ can be written as  
\begin{align}
\Ecal[\rho] &= \sum_{\alpha=1}^D \tr_{\mathrm{in}}\left[\left(\openone_{\mathrm{out}} \otimes \rho^{\mathrm{T}}\right)\lambda_\alpha \ketbra{\alpha}{\alpha}\right]   \\
\label{eqn::HermPres}
 &=\sum_{\alpha=1}^D \lambda_\alpha \,\left(\sum_{l=1}^{d_\mathrm{in}} \braket{l|\alpha}\!\bra{l}\right) \rho \left( \sum_{k=1}^{d_\mathrm{in}}\ket{k}\! \braket{\alpha|k}\right)
\equiv \sum_{\alpha=1}^D \lambda_\alpha\, \widetilde{K}_\alpha \rho \widetilde{K}_\alpha^{\dagger}\, ,
\end{align}
which implies that the matrices $\{L_\alpha,R_\alpha\}$ of the map's operator sum representation satisfy $L_\alpha = \pm R_\alpha \ \forall \alpha$.
In fact, the ability to write the map's action in the form of Eq.~\eqref{eqn::HermPres} is a necessary and sufficient condition for Hermiticity preservation~\cite{verstraete_quantum_2002, jordan:052110}, as is the Hermiticity of its B form $\Ecal_B$.

\subsubsection{Complete positivity}
\label{subsec::CP}

In addition to being Hermiticity and trace preserving, a physical quantum map $\Ecal$ must be positive, \textit{i.e.}, it must map positive matrices $\rho$ to positive matrices $\rho'$. What is more, it must be \emph{completely positive} (CP): \textit{any} trivial extension $\Ecal\otimes \Ical_{\an}: \mathcal{B}(\Hcal_{d}) \otimes  \mathcal{B}(\Hcal_{n_\an}) \rightarrow \mathcal{B}(\Hcal_{d}) \otimes  \mathcal{B}(\Hcal_{n_\an})$, where $\Hcal_{n_\an}$ is $n$-dimensional and $\Ical_\an$ is the identity map on $\mathcal{B}(\Hcal_{n_\an})$, must also be positive. In other words, a meaningful operation $\Ecal$ that acts non-trivially only on a subset of the degrees of freedom of a quantum state should not yield a non-physical result: $(\Ecal \otimes \Ical_\an)[\eta] \geq 0$ for any $\eta \geq 0$ and any $n \geq 1$.

The above justification for CP is an operational one. However, CP also guarantees that the action of the (trace preserving) map comes from a joint unitary dynamics of the system with an environment, as proven by Stinespring in 1955~\cite{stinespring1955} (see Fig.~\ref{fig::InputOutput}). This result was recently generalised for trace non-increasing CP maps~\cite{chiribella_transforming_2008}, which can, in principle, be physically realised within quantum mechanics as joint unitary dynamics with postselection.

If the map $\Ecal$ is CP, than its B form $\Ecal_B = \Ecal \otimes \Ical\left[\ketbra{I}{I}\right]$ is non-negative by definition\footnote{It is clear from this that CP also implies Hermiticity preservation.}, and hence all its eigenvalues satisfy $\lambda_\alpha \geq 0$. This allows Eq.~\eqref{eqn::HermPres} to be further simplified:
\begin{gather}
\label{eqn::canonical}
\Ecal[\rho] = \sum_{\alpha=1}^D \left(\sum_{l=1}^{d_\mathrm{in}}\sqrt{\lambda_\alpha} \braket{l|\alpha}\!\bra{l}\right) \rho \left(\sqrt{\lambda_\alpha} \sum_{k=1}^{d_\mathrm{in}}\ket{k}\! \braket{\alpha|k}\right) \equiv \sum_{\alpha = 1}^D K_\alpha \rho K_\alpha^{\dagger}\, ,
\end{gather}
This form of the map was first noticed by Sudarshan \textit{et al.}~\cite{SudarshanMatthewsRau61} in 1961 by means of eigendecomposition of $\Ecal_B$. However, it is now commonly referred to as the Kraus form~\cite{Nielsen00a} and the $d_{\mathrm{out}}\times d_\mathrm{in}$ matrices $\{K_\alpha\}$ are called the \emph{Kraus operators} of $\Ecal$~\cite{kraus_general_1971,kraus_states_1983}. CP therefore implies that $L_\alpha = R_\alpha, \, \forall \alpha$ in the general operator sum representation of $\Ecal$.
 
\textbf{Properties of the Kraus form.}
As in the case of general linear maps, the set of Kraus operators that corresponds to $\Ecal$ is not unique. Any set $\{K'_\mu\}$ of $d_{\mathrm{out}}\times d_\mathrm{in}$ matrices that is related to $\{K_\alpha\}$ by an isometry gives rise to the same map, \textit{i.e.}, $\{K'_\mu = \sum_{\alpha'} (U)_{\mu\alpha'}K_{\alpha'}\}$, where $U^{\dagger}U = \openone$, is also a valid set of Kraus operators for $\Ecal$.
The minimal number of operators needed for the operator sum representation of a CP map $\Ecal$ is called its \emph{Kraus rank}. It coincides with the rank of the Choi state $\Upsilon_{\Ecal}$~\cite{verstraete_quantum_2002}.

Every CP map allows for a \emph{canonical} Kraus decomposition,  where the number of Kraus operators is minimal and they are mutually orthogonal, \textit{i.e.}, $\tr(K_\alpha K^{\dagger}_{\alpha'}) \propto \delta_{\alpha \alpha'}$. In fact, the Kraus decomposition derived in~\eqref{eqn::canonical} is already canonical:
\begin{gather}
\tr(K_{\alpha}K_{\alpha'}^{\dagger}) = \sum_{k,l} \sqrt{\lambda_\alpha \lambda_{\alpha'}} \tr( \braket{l|\alpha}\!\braket{l|k} \!\braket{\alpha'|k} ) = \lambda_{\alpha}\delta_{\alpha\alpha'}\,.
\end{gather}

So far, we have shown that the action of a CP map can be expressed in terms of a Kraus decomposition. The inverse of this statement is also true. If the action of a map $\Ecal$ can be written as $\Ecal[\rho] = \sum_{\alpha}K_\alpha\rho K_{\alpha}^{\dagger}$, then 
\begin{gather}
\label{eqn::KrausProof}
\left(\Ecal \otimes \Ical_\an \right)[\eta] = \sum_{\alpha}(K_\alpha \otimes \openone)\eta (K_\alpha^{\dagger} \otimes \openone) = \sum_{\alpha}[(K_\alpha \otimes \openone)\sqrt{\eta}\,][\sqrt{\eta}\, (K_\alpha^{\dagger} \otimes \openone)]\,,
\end{gather}
where $\sqrt{\eta}$ exists and is positive due to the positivity of $\eta$. The last term in~\eqref{eqn::KrausProof} is of the form $\sum_\alpha A_\alpha A_\alpha^{\dagger}$ which is positive, as every term $A_\alpha   A_\alpha^{\dagger}$ is positive on its own. As this is true independent of the size of $\eta$, we have shown that $\Ecal \otimes \Ical_\an$ is positive if the action of $\Ecal$ can be written in terms of a Kraus decomposition. This means that a map $\Ecal$ is CP iff its action can be written in terms of a Kraus decomposition. Equivalently, $\Ecal$ is CP iff $\Ecal_B$ is positive, which implies that a map for which $\Ecal \otimes \Ical_\an \geq 0,\, \forall \ \mathrm{dim}(\Hcal_{n_\an}) \le d$ satisfies $\Ecal \otimes \Ical_\an \geq 0$, for any dimension of $\mathrm{dim}(\Hcal_{n_\an})$.

\begin{figure}
\centering
\includegraphics[scale=1.1]
{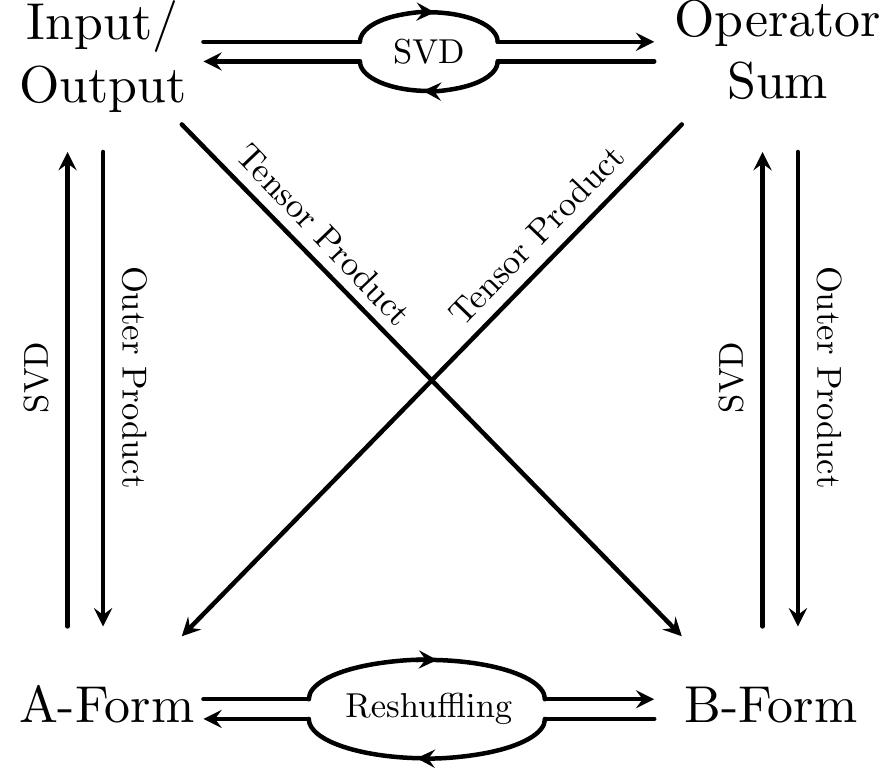}
\caption{\emph{Converting between different representations.} Even though we have not drawn the corresponding arrows explicitly, it is, \textit{e.g.}, possible to get from the A form to the operator sum representation by reshuffling followed by a singular value decomposition (SVD) (analogously for B form to input/output representation).}
\label{fig::CommutationDiag}
\end{figure}

\subsection{Representations of quantum maps -- a summary}
\label{sec::SummRep}

All the representations introduced above constitute different concrete ways of expressing the action of the same abstract map $\Ecal$. This situation is reminiscent of differential geometry, where (abstract) geometrical objects can be expressed in terms of different coordinate systems. And, just like in differential geometry, where ``in practice few things are more useful than a well-chosen coordinate system''~\cite{bengtsson_geometry_2007}, which representation of $\Ecal$ is most advantageous depends on the respective experimental or computational context. 

While the A form does not possess particularly nice mathematical properties, even for the case of quantum maps, the fact that its action can be written in terms of a simple matrix multiplication, makes it appealing for numerical simulations (where differential equations must often be expressed in vector form). On the other hand, the properties of the B form make it easy to check whether or not a map corresponds to a physical process. It also embodies the CJI between quantum maps and quantum states, which can be used for the ancilla-assisted tomography of quantum maps. 

The tomographic representation of $\Ecal$ is closest in spirit to the experimental reconstruction of the action of $\Ecal$. Given the experimentally obtained input/output relation between a basis of inputs and their corresponding outputs, it allows one to directly infer the action of $\Ecal$ on an arbitrary input state.

Lastly, the operator sum representation is particularly advantageous from a theoretical point of view. Proving that the dynamics of an open system for a particular initial state is CP amounts to showing that it can be written in terms of a Kraus decomposition~\cite{kraus_states_1983, alicki_semi_1987, rodriguez-rosario_completely_2008, pollock_complete_2015}, and the existence of a minimal Kraus decomposition can be employed to show the existence of generalised Stinespring dilations~\cite{chiribella_transforming_2008, chiribella_theoretical_2009}.

The above list of applications of different representations is by no means exhaustive, but it gives a flavour of when they are each most useful. We have summarised the different representations and their properties in Table~\ref{tab::Summary}, while Fig.~\ref{fig::CommutationDiag} depicts how to convert from one representation to another. 

\begin{table}
\setlength{\tabcolsep}{3pt}
\renewcommand{\arraystretch}{1.3}
\centering
\begin{tabular}{c||c|c|c|c}
& \makecell{Kraus  \vspace{-4pt}\\ (Operator Sum)}  & \makecell{Input/Output \vspace{-4pt} \\  (Tomographic)} & \makecell{Sudarshan B form  \vspace{-4pt}\\ (Choi Matrix)} & \makecell{Sudarshan A Form } \\\hline \hline
Rep. &
$\{L_\alpha,R_\alpha\}$ & $\{\hat{\rho}_i',\hat{D}_i\}$ & $ \phantom{\begin{array}{l}
. \\ .\end{array}} \Ecal_B = \begin{cases}
 \sum_i \hat{\rho}_i' \otimes \hat{D}_i^*\\ \sum_\alpha L_\alpha \times R_\alpha^* \end{cases}$&  $ \Ecal_A = \begin{cases}\sum_i \hat{\rho}_i' \times \hat{D}_i^* \\ \sum_\alpha L_\alpha \otimes R_\alpha^*\end{cases}$    \\ \hline

Action  &
$\rho' =\sum_\alpha L_\alpha \rho \, R_\alpha^{\dagger}$ & $\rho'=\sum_i \hat{\rho}_i'\tr(\hat{D}_i^{\dagger}\rho)$ &  $\rho' = \tr_{\mathrm{in}}[(\openone_{\mathrm{out}} \otimes \rho^{\mathrm{T}})\, \Ecal_B]$  &   $\ket{\rho'}\rangle = \Ecal_A\ket{\rho}\rangle$ \\ 
\hline
TP &
$\sum_\alpha R_\alpha^{\dagger} L_\alpha = \openone$ & $\sum_i\tr(\hat{\rho}_i')\hat{D}_i^{\dagger} = \openone$  &  $\tr_{\mathrm{out}}(\Ecal_B) = \openone_\mathrm{in} $  & \multirow{4}{*}{\diagbox[dir=NE]{\hspace{15pt} }{\\ \hspace{30pt}} }\\ \cline{1-4}

HP &$L_\alpha = \pm R_\alpha \ \ \forall \alpha$ & $\begin{array}{@{}r@{}l@{\quad}} \hat{\rho}_i^{\dagger}&{}= \hat{\rho}_i \ \ \forall i \\ \Rightarrow \ (\hat{\rho}_i')^\dagger&{}= \hat{\rho}_i' \ \ \forall i \end{array}$ & $\Ecal_B^{\dagger} = \Ecal_B$ &  \\ \cline{1-4}

CP & $L_\alpha = R_\alpha \ \ \forall \alpha$ & $\sum_i \hat{\rho}_i' \otimes \hat{D}_i^* \geq 0$ & $ \Ecal_B \geq 0 $ & \\ \hline

\end{tabular}
\caption{\emph{Linear maps in different representations.} Note that the A form does not possess particularly nice properties for trace preserving (TP), Hermiticity preserving (HP) or completely positive (CP) maps. Hermiticity preservation for the Input/Output representation is denoted only for the case where all inputs $\hat{\rho}_i$ are Hermitian.}
\label{tab::Summary}
\end{table}

\section{Generalisations of quantum maps}
\label{sec::beyond}

The quantum maps described in the previous section take the initial state $\rho$ of the system at a particular point in time $t_0$ to that at a particular later time $t$. Consequently, they allow for the calculation of two-time correlation functions between observables. Their experimental reconstruction, as introduced in Sec.~\ref{subsec::Forms}, is well-defined if the relation between input and output states is linear; this, in turn, means that the system can be prepared independently of its environment. We will see below that this implies that the system and environment are in a product state $\rho\otimes \tau_{\e}$ at $t_0$. If the experimental situation is such that the initial system state is correlated with the environment, or \emph{multi}-time correlation functions are of interest, quantum maps from density operators to density operators are neither well-defined nor sufficient as a description of the experimental situation. 

We first discuss the problem of initial correlations, and various attempts to solve it. We will then offer an operational resolution, which opens up the door to describe arbitrary quantum processes.

\subsection{Initial correlation problem}
\label{sec::IC}

In the late 1990s and early 2000s, experimentalists began reconstructing quantum gates -- the fundamental elements of a quantum computer -- by means of quantum process tomography~\cite{Nielsen:1998py, PhysRevA.64.012314, PhysRevLett.91.120402, Wein:121.13, orien:080502, NeeleyNature, chow:090502, Howard06, myrskog:013615}. Ideally a quantum gate is a unitary operation, but in practice they can be noisy. Therefore, the results of these experimental reconstructions were expected to be CP quantum maps. Yet, to the surprise of many researchers, the reconstructed maps were often not CP, and it was not clear why. 

This initiated a flurry of theoretical explanations, one of which suggested that, if the initial state of the system is correlated with its environment, the quantum map describing the dynamics of the system need not be completely positive~\cite{shaji_whos_2005}. As mentioned above, Stinespring's theorem~\cite{stinespring1955} guarantees that any CP dynamics for the system $\s$ can be thought of as coming from unitary dynamics of the system with an environment $\e$. However, this construction assumes that the initial state of the system-environment ($\se$) state is uncorrelated -- a very restrictive assumption in practice. For instance, consider the case where the initial $\se$ state at $t_0$ is uncorrelated, meaning that the dynamics to some later $t_1$ is CP. In general, at $t_1$ the state of $\se$ will be correlated, and if we want to describe the dynamics from $t_1$ to later $t_2$, the quantum maps discussed in Sec.~\ref{sec::QuantumMaps} no longer apply.

\subsubsection{Not completely positive maps, not completely useful}
\label{sec::NCP}

As most clearly elucidated by Pechukas in his seminal paper~\cite{pechukas94a} (and in a subsequent exchange between him and Alicki~\cite{Alicki95, Pechukas95}), a map whose argument is the state of $\s$ is both completely positive and linear iff there are no initial $\se$ correlations. Pechukas originally proved the theorem for qubits, but it was later generalised to $d$-dimensional systems in Ref.~\cite{jordan:052110}; here, we give a version of this result that closely resembles Ref.~\cite{PhysRevA.81.012313}.

Pechukas introduced an \textit{assignment map}  $\mathfrak{A}: \mathcal{B}(\Hcal_{d_\s}) \rightarrow \mathcal{B}(\Hcal_{d_\s} \otimes\Hcal_{d_\e})$, which assigns a $\se$ operator for every $\s$ state, with a consistency condition: $\tr_\e\mathfrak{A}[\rho] = \rho \ \ \forall \rho$. Concatenating the assignment map $\mathfrak{A}$ with a unitary $U_\se$, and tracing over $\e$, gives a map $\Fcal$:
\begin{gather}
\Fcal[\rho] = \tr_\e\left\{U_\se \mathfrak{A}[\rho] U^\dagger_\se \right\} \equiv \tr_\e\left\{U_\se \rho^0_\se U^\dagger_\se \right\}. \label{eqn::dilation}
\end{gather}
The unitary $U_\se$ and trace over $\e$ are both CP maps; therefore, if we require that $\mathfrak{A}$ is linear and CP, then it follows that $\Fcal$ must also have these properties (and is therefore a legitimate quantum map of the sort discussed in the previous section).  

Now, for a consistent and CP assignment it follows that, for a basis $\{\hat{\rho}_i\}$ consisting of pure states, $\mathfrak{A}[\hat{\rho_i}] = \hat{\rho_i}\otimes {\tau_\e}_i$, where ${\tau_\e}_i$ have to be density operators (as required for positivity of the assignment). By the same argument, the action of the assignment map must also give a product $\se$ state on \textit{any} pure state. Let us take a pure state not in the basis and linearly express it as $\sigma = \sum_i c_i \hat{\rho_i}$, where $c_i$ are real, with $\sum_i c_i =1$, but not necessarily positive. The action of the assignment map gives us $\sum_i c_i \hat{\rho_i}\otimes {\tau_\e}_i = \sum_i c_i \hat{\rho_i}\otimes {\tau_\e}$ and, therefore, ${\tau_\e}_i=\tau_\e$~$\forall i$. That is, the initial $\se$ state is uncorrelated: $\mathfrak{A}[\rho]=\rho\otimes\tau_\e$. Conversely, if the initial $\se$ state $\rho_\se^0$ is correlated, then either complete positivity or linearity must be abandoned\footnote{Another option is to give up consistency, but this too is not desirable. We will address this matter in some detail at the end of this subsection.}~\cite{Alicki95, Pechukas95}. 

\begin{figure}
\centering
\includegraphics[scale=1.7]
{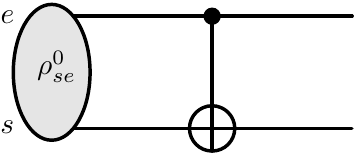}
\caption{A simple circuit, for which the reduced dynamics of $\s$ is not describable by a CP quantum map when $\rho^0_\se$ is initially correlated.}
\label{fig::NCPexample}
\end{figure}

From an experimental standpoint, this state of affairs is problematic. On the one hand, complete positivity is a useful property -- giving up CP means giving up the Holevo quantity~\cite{holevo-assingment}, data processing inequality~\cite{buscemi_complete_2014}, and entropy production inequality~\cite{argentieri2014violations} -- and a CP description naturally predicts the physical fact that one always reconstructs positive probabilities (even for correlated preparations). On the other hand, dropping linearity is not a viable option either: complete tomography is not possible when the dynamics is nonlinear -- at least not in a finite number of experiments.

Giving up either is undesirable; however, faced with this choice, many researchers have opted to relinquish complete positivity of dynamics in favour of a framework for open dynamics based on not-completely positive (\textbf{NCP}) maps~\cite{StelmachovicBuzek01, jordan:052110}. In brief, NCP maps $\Fcal^{\rm NCP}$ are linear maps that preserve positivity for some subset of the space of system density operators, but fail to do so on the remaining set. They take as their starting point an assignment map $\mathfrak{A}$, as above, but do not require it to produce a positive $\se$ operator for all inputs. Instead, $\mathfrak{A}$ is required to be consistent, such that $\mathfrak{A}[\tr_\e\rho_\se^0] = \rho_\se^0$ for some correlated $\rho_\se^0$, and the action of $\Fcal^{\rm NCP}$ is only defined on the set $\{\rho : \mathfrak{A}[\rho]\geq0\}$ (which always contains $\tr_\e\rho_\se^0$), called the compatibility domain of the map. Its action can then be defined through the dilation in Eq.~\eqref{eqn::dilation}, which will only result in a positive output when $\rho$ is in the compatibility domain.

While mathematically well-defined (though not unique), the NCP framework lacks a clear link to the operational reality of quantum dynamics. It assumes that there is a family of initial system states (the compatibility domain) available, and that the experimenter knows exactly which of these states is the input in each run of the experiment.  However, unlike in a classical stochastic process -- where an experimenter can observe initial and final states of a system without disturbing it -- there is no operational mechanism for identifying which initial state $\rho$ will undergo the evolution in the quantum case. That is, there is no way for the experimentalist to differentiate between initial states in any given run without disturbing the system, and hence changing the correlations between $\s$ and $\e$. Without such a mechanism, the concept of a compatibility domain becomes a purely mathematical notion, void of physical meaning. Instead, the experimenter is presented in each run with a fixed average $\se$ state, which they can then prepare by performing a \textit{control operation}.

In other words, there is no unambiguous way to go into the laboratory and directly reconstruct a NCP map through process tomography. In fact, if one were to attempt such a reconstruction, it would quickly become apparent that the dynamics depends not on the initial state, but on how that state is prepared. To see this more clearly, consider the two-qubit circuit in Fig.~\ref{fig::NCPexample} with initially correlated pure state $\rho^0_\se = \ketbra{\psi}{\psi}$, where $\ket{\psi} = \frac{1}{\sqrt{2}} (\mu \ket{00} + \nu \ket{11})$, and an $\s\e$ dynamics given by a CNOT gate. An experimenter wishing to tomographically reconstruct the dynamics of just one of the qubits would have to prepare a variety of initial system states (we will return to this point later). Say they intended to prepare the initial state $\ketbra{0}{0}$. This would involve making a projective measurement (for sake of argument, in the computational basis $\{\ket{0},\ket{1}\}$) followed by a unitary transformation that depends on the outcome -- $\openone$ in the case that the outcome corresponding to projector $\Pi_0=\ketbra{0}{0}$ is observed, and $\sigma_x$ in the case that the outcome corresponding to $\Pi_1=\ketbra{1}{1}$ is observed. However, the state of the environment qubit, and hence the subsequent dynamics, would also depend on the measurement outcome: $\tau_{\e|0} = \tr_\s\{\Pi_0\otimes\openone\ketbra{\psi}{\psi}\} = \ketbra{0}{0} \neq \tau_{\e|1} = \tr_\s\{\Pi_1\otimes\openone \ketbra{\psi}{\psi}\} = \ketbra{1}{1}$. That is, despite the fact that the initial state is the same in these two cases, the final density operator differs (in fact, the two different output states are orthogonal). Choosing a different preparation procedure will not alleviate these issues, and similar problems arise for any initially correlated state. This leads us to conclude that there is no unique way to prepare a state, and the preparation procedure plays a role in determining the future evolution~\cite{modi_preparation_2011}.

Let us take this argument one step further and attempt to perform quantum process tomography by preparing the basis states with projections. For simplicity, we will confine the tomography to the $x-z$ plane of the Bloch sphere. We prepare basis states $\Pi_0$ and $\Pi_1$ by projecting the system in the $z$ basis, and basis state $\Pi_+ = \ketbra{+}{+}$ by projecting in the system in the $x$ basis. In the latter case, sometimes we will find the system in state $\Pi_- = \ketbra{-}{-}$, which is linearly related to the basis states as $\Pi_- = \Pi_0 + \Pi_1 - \Pi_+$. The output states corresponding to basis states $\Pi_0, \Pi_1, \Pi_+$ are easily computed to be $\Pi_0, \Pi_0, \Pi_+$. And similarly, by examining the global dynamics we find that $\Pi_-$ maps to itself. However, if we try to construct a linear map using the input/output data, we find that it predicts $\Pi_-$ will be mapped to matrix $\Pi_0 + \Pi_0 - \Pi_+$, which is non-positive. The constructed map is therefore NCP, and it makes a nonsensical prediction. Clearly, we can prepare $\Pi_-$ and observe the subsequent output of the process. But the constructed map does not capture this physics, and we have to conclude that NCP maps are not useful. Indeed, there are many examples where quantum process tomography, without properly taking the preparation procedure into account, leads to NCP and nonlinear maps\footnote{On the other hand, it is possible to construct a meaningful map where all preparations are projective~\cite{kuah:042113}, or with any other restricted set of preparations~\cite{milz_reconstructing_2016}, when these are correctly accounted for.}~\cite{modi_role_2010}.

Before introducing a resolution to the problem of initial correlations, we discuss the matter of giving up consistency to retain CP and linearity of the dynamics. Research along these lines led to the claim that ``vanishing quantum discord is necessary and sufficient for completely positive maps"~\cite{PhysRevLett.102.100402} which received a great deal of attention, but then was subsequently proven to be  incorrect~\cite{PhysRevA.87.042301}, leading to an erratum~\cite{PhysRevLett.116.049901}. In Ref.~\cite{rodriguez-rosario_completely_2008}, it was shown that if the initial $\se$ state has vanishing quantum discord, then a CP map can be ascribed to the dynamics of $\s$. Consequently, by projectively measuring the system part of any initial state $\rho^0_\se$ -- which will always produce a discord zero state -- one can associate a CP map from the measurement outcome at the initial time to the quantum state at the final time. The problem with this approach is that the CP maps depend on the choice of measurement, which does not depend on the pre-measurement state of the system. The corresponding assignment map is therefore not consistent, and one is left with only a partial description of the dynamics (with similar issues to the example above).

\subsubsection{Operational resolution: Superchannels}

As already mentioned, the first step of any experiment is to prepare the system in a desired state by applying a control operation. The control operations can be anything, including unitary transformations, projective 
measurements, projective measurements followed by a unitary transformation (like in the example above) and everything in-between.  Mathematically, a control operation $\Acal_\s$ is just a (trace non-increasing) CP quantum map (as described in Sec.~\ref{sec::QuantumMaps}). In a dilated picture the final state is related to the control operation as the following:
\begin{gather}\label{eqn::SCD}
\rho'  = \tr_\e \{U (\Acal_\s \otimes \Ical_\e [\rho^0_\se]) U^\dag \} \equiv \Mcal[\Acal_\s]\, ,
\end{gather}
where we have defined the \emph{superchannel} $\mathcal{M}$~\cite{modi_operational_2012}, a linear operator that maps preparations to final states. From here on we omit the subscript $\s$ on the control operation $\Acal_\s$, and assume it only acts on the system. 

From Eq.~\eqref{eqn::SCD}, it is clear that the superchannel is linear in the same way as $\Ecal$, as argued at the beginning of Sec.~\ref{sec::QuantumMaps}. The set of all control operations is isomorphic to the set of positive $d^2\times d^2$ matrices of trace less than or equal to $d$ (the B forms of control operations). Henceforth, whenever we write $\Acal$, we always mean this representation of the control operation, if not explicitly stated otherwise. With this, by taking the square root of the initial state and combining it with $U$ we can write the action of $\mathcal{M}$
\begin{gather}
\label{eqn::CPsuper}
\quad \Mcal[\Acal] = \sum_\alpha \mu_\alpha \ \Acal \ \mu_\alpha^\dag \qquad \mbox{with} \qquad \mu_{\epsilon x}= \sqrt{\lambda_x} \bra{\epsilon} U \otimes_\s \ket{\Psi_x}^{\mathrm{T}_\s} \qquad \mbox{and} \qquad \alpha = \epsilon x.
\end{gather}
Here $\lambda_x$ and $\ket{\Psi_x}$ are the eigenvalues and eigenectors of $\rho^0_\se$ respectively and $\otimes_\s$ and $\mathrm{T}_\s$  means a tensor product and transpose (in computational basis) on the space of $\s$ only respectively, while the normal matrix product applies to the space of $\e$. The last equation is an analogue of Eq.~\eqref{eqn::canonical}; it is the Kraus representation for the superchannel, which means that it is CP~\cite{modi_operational_2012}. In fact, the operators $\mu_\alpha$ have similar properties to those in Sec.~\ref{subsec::CP} and it is straightforward to show that $\Mcal$ is trace preserving in the sense that it maps trace preserving preparations to unit trace matrices. From a mathematical point of view, the superchannel is a CP map just as the ones encountered in Sec.~\ref{sec::QuantumMaps}, but with input and output spaces of different size, \textit{i.e.}, $\Mcal:  \mathcal{B}(\Hcal) \otimes \mathcal{B}(\Hcal) \rightarrow \mathcal{B}(\Hcal)$. The CP nature of the superchannel also has an operational implication: Suppose we bring in an auxiliary system $\an$ of dimension $n$ and perform an entangling control operation $\Acal_{\s \an}$, before letting $\s$ undergo the process in question (\textit{i.e.}, interact with $\e$). The complete positivity of the superchannel guarantees that the final state $\rho'_{\s\an}$ will always be positive.

Operationally speaking, the superchannel is simply the logical consequence of the input/output picture presented at the very beginning of the paper; it maps the actual controllable inputs (the preparations $\Acal$) to the actual measurable outputs (the final system state $\rho'$) of the experiment. When there are no initial correlations, \textit{i.e.}, the initial $\se$ state in Eq.~\eqref{eqn::SCD} is a product state we find that the Kraus operators in Eq.~\eqref{eqn::CPsuper} become $\mu_{\epsilon x} = K_\epsilon \otimes \sqrt{\lambda_x^\s} \bra{\psi_x^\s}^*$, where $K_\epsilon$ are the Kraus operators of the quantum map $\Ecal$ in Eq.~\eqref{eqn::canonical}, $\lambda_x^\s$ and $\ket{\psi_x^\s}$ are eigenvalues and eigenvectors of $\rho^0_\s$ respectively. Subsequently, the action of the superchannel reduces to $\Mcal[\Acal] = \Ecal(\rho_\s)$, where $\rho_\s$ is the result of applying the control operation $\Acal$ on the fiducial initial state $\rho_\s^0$, \textit{i.e.},  $\rho_\s \equiv \Acal[\rho^0_\s]$. Consequently, the superchannel formalism includes the experimental situation depicted in the first chapter and naturally extends quantum maps to the more general case of initial correlations.

Proponents of the NCP map formalism would claim that the superchannel framework (and our experimenter in the example in the previous subsection) is setting up a different dynamical experiment each time the system is prepared. However, the superchannel only depends on the initial $\se$ state and subsequent $\se$ unitary operation; it is independent of the choice of the control operation, which is the choice of the experimenter. Moreover, the superchannel, along with some data processing, contains all NCP maps one could construct for the process (though still void of any operational meaning)~\cite{modi_operational_2012}. Conversely, the only way the predictions of the superchannel could be reproduced in the NCP formalism is by enumerating the NCP maps corresponding to all (infinitely many) possible preparation procedures.

The construction of the superchannel does not \textit{a priori} assume linearity or complete positivity. However, by simply following the operational reality of a quantum experiment, we have arrived at a map that has these familiar (and desirable) features. In doing so, we have overcome Pechukas' theorem; the superchannel is a consistent, linear, and CP description for dynamics in the presence of initial correlations. Unlike NCP maps, it has a clear operational meaning, and has been unambiguously reconstructed in a tomography experiment~\cite{PhysRevLett.114.090402}. Finally, the CP nature of the superchannel allows for the extension of useful results, such as the Holevo quantity, data processing inequality, and entropy production inequalities, to the case of initial correlations~\cite{PhysRevA.92.052310}.

We will now show how the superchannel concept can be generalised to processes involving multiple time steps, before discussing its structure and representations.

\subsection{Multiple time steps and the process tensor}

Like the quantum maps $\Ecal$ from Sec.~\ref{sec::QuantumMaps}, the superchannel only accounts for two time correlations between preparations at the initial time and measurements at the final time. In a more general experiment -- for example, in a multi-dimensional spectroscopy experiment~\cite{KassalBook} -- one may want to know about correlations across multiple time steps. It is relatively straightforward to generalise the superchannel to this scenario; imagine that the experimenter performs (CP) control operations $\Acal_0, \Acal_1,\dots,\Acal_{k-1}$ at the $k$ times $t_0, t_1,\dots, t_{k-1}$ and measures the corresponding output state $\rho'_k$\footnote{If the control operations are not trace preserving (and therefore not performable deterministically), then the trace of $\rho_k'$ gives the probability of performing those operations.} at $t_k$. This scenario is illustrated in Fig.~\ref{fig::process-a} (our only assumption is that these operations can be performed on a much shorter time scale than any other dynamics of $\s$ or $\se$). This setup is very general; for instance, the control operations could be quantum gates, with the final state corresponding to the outcome of a quantum computation. Or, perhaps, the process could be a series of chemical reactions, where the control operations represent the addition of reactants.

\begin{figure}
\centering
\includegraphics[scale=1.5]
{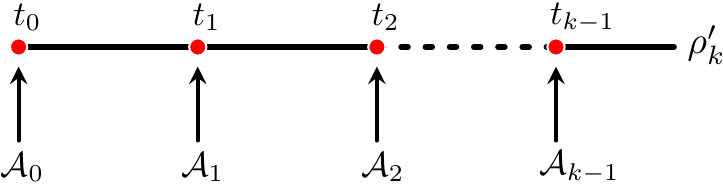}
\caption{\emph{A $k$-step process.} At each time step $t_i$, a CP operation $\Acal_i$ is performed, and the resulting state $\rho'_k$ at $t_k$ is determined by quantum state tomography. The scenario in which dynamics is described by the quantum maps of Sec.~\ref{sec::QuantumMaps} is also included in this schematic; it corresponds to a single-step process, where $\Acal_0$ is the preparation of the input state.}
\label{fig::process-a}
\end{figure}

Just as the superchannel is linear in its argument, the final state $\rho'_k$ depends in a \emph{multilinear} way on the operations $\Acal_j$. Mathematically, this means that the dynamics is a mapping $\Tcal^{k:0}: \mathcal{B}(\mathcal{B}(\Hcal))^{\otimes k}\rightarrow \mathcal{B}(\Hcal)$, called a \emph{process tensor}~\cite{pollock_complete_2015}, whose action can be written as 
\begin{gather}
\label{eqn::processtensor}
\rho_k' = \Tcal^{k:0}\left[\mathbf{A}_{k-1:0} \right]\, .
\end{gather}
In terms of their B form, we have $\mathbf{A}_{k-1:0}\in \left[\mathcal{B}(\Hcal)\otimes \mathcal{B}(\Hcal)\right]^{\otimes k}$  and $\Tcal^{k:0}$ becomes a mapping $\Tcal^{k:0}: \left[\mathcal{B}(\Hcal) \otimes \mathcal{B}(\Hcal)\right]^{\otimes k} \rightarrow \mathcal{B}(\Hcal)$. 
To keep better track of the different terms, we give superscripts to the process and subscripts to the control operations. For the case of independent control operations, $\mathbf{A}_{k-1:0}$ is simply given by $\mathbf{A}_{k-1:0}=\Acal_{k-1}\otimes\cdots \Acal_1\otimes \Acal_0$. In a more general scenario, the operations could be correlated, either classically (\textit{e.g.}, transformations conditioned on earlier measurement outcomes) or quantum mechanically, through successive interactions with an ancilla.

In Refs.~\cite{chiribella_theoretical_2009,pollock_complete_2015} the existence of a \emph{generalised Stinespring dilation} was proven; a map $\mathcal{T}^{k:0}$ is consistent with $\se$ unitary dynamics if it is linear, CP, trace preserving in the sense that it maps sequences of trace preserving control operations to unit trace matrices, and possesses a containment property, $\Tcal^{j:k} \subset \Tcal^{i:l} \, \forall \, i \le j \le k \le l$. The latter property is a causality property ensuring that future actions do not affect past dynamics. Conversely, the process tensor can be derived starting from a dilated (unitary) $\se$ evolution, as shown in Fig.~\ref{fig::Process-b}. We sketch this derivation now.

\begin{figure}
\centering
\includegraphics[scale=1.1]
{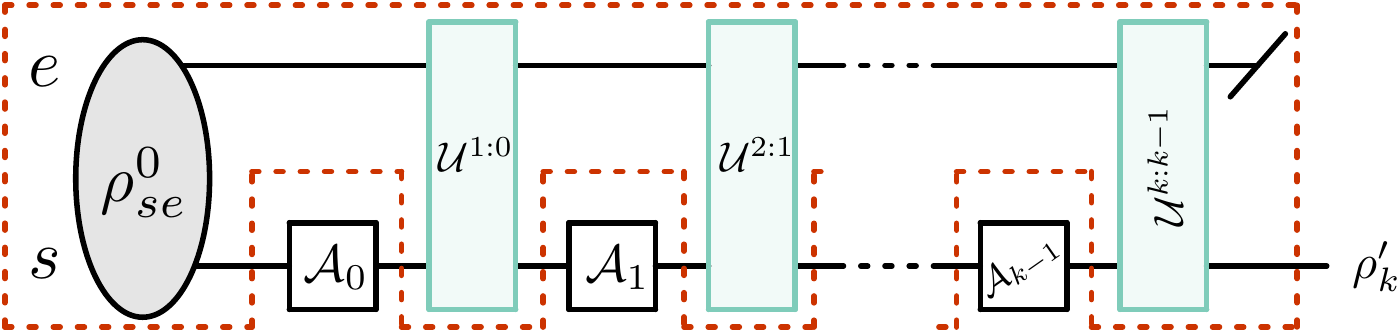}
\caption{
\textit{Generalised Stinespring dilation.} Any $k$-step process tensor has a dilated representation, \textit{i.e.}, there exists a set of unitary maps $\{\mathcal{U}^{1:0},\dots,\mathcal{U}^{k:k-1}\}$ and an initial system environment state $\rho^0_{se}$, such that $\rho_k'=\Tcal^{k:0}[\Acal_{k-1} \otimes \cdots \otimes \Acal_0] = \tr_{e}\left\{\mathcal{U}^{k:k-1}\left(\Acal_{k-1}\otimes \Ical_e \right) \cdots \mathcal{U}^{1:0}\left(\Acal_0\otimes \Ical_e\right)[\rho^0_{se}]\right\}$. The process tensor $\Tcal^{k:0}$ corresponds to `everything that cannot be controlled by the experimenter', \textit{i.e.}, the area framed by the orange dotted lines in the figure.}
\label{fig::Process-b}
\end{figure}

Any $\se$ dynamics can be written as
\begin{gather}
\label{eqn::processdilated}
\rho'_k = \tr_\e \left[\Ucal^{k:k-1} \Acal_{k-1} \ \Ucal^{k-1:k-2} \Acal_{k-2} \dots \Ucal^{1:0} \Acal_0 (\rho^0_\se) \right],
\end{gather}
where $\Acal$ act on $\s$ alone and the unitary maps $\Ucal^{l:k} (\rho^{k}_\se) = U^{l:k} \rho^{k}_\se {U^{l:k}}^\dag= \rho^l_\se$ act on the full system environment space. Everything in this equation other than the control operations, \textit{i.e.}, everything in the red box in Fig.~\ref{fig::Process-b}, can be considered as part of the process. In analogy to Eq.~\eqref{eqn::CPsuper}, by contracting (`matrix multiplying') the unitary operators in the space of $\e$, along with the initial state $\se$ and taking the final trace, we can define
\begin{gather}
T_{\epsilon x}^{k:0} \equiv \sqrt{\lambda_x} \bra{\epsilon} U^{k:k-1} \otimes_\s \bbra{U^{k-1:k-2}}_\s \otimes_\s \dots \otimes_\s \bbra{U^{1:0}}_\s \otimes_\s \ket{\Psi_x}^{\mathrm{T}_\s},
\end{gather}
where again $\lambda_x$ and $\ket{\Psi_x}$ are the eigenvalues and eigenvectors of $\rho^0_\se$ respectively, $\otimes_\s$, $\bbra{\;}_s$, and $\mathrm{T}_\s$ mean tensor product, vectorisation (in the sense of Eq.~\eqref{eqn::VecAction}), and transpose on the space of $\s$ only. Note that the last unitary matrix is not vectorised. With this, we can rewrite Eq.~\eqref{eqn::processdilated} as
\begin{gather}
\label{eqn::processdilated2}
\rho'_k =
\sum_\alpha T_\alpha^{k:0} \ \mathbf{A}_{k-1:0} \
{T_\alpha^{k:0}}^\dag = \Tcal^{k:0}[\mathbf{A}_{k-1:0}], \quad \text{with} \quad \alpha = \epsilon x\, ,
\end{gather}
where $\Tcal^{k:0}$ is the process tensor. This equation is an analogue of Eqs.~\eqref{eqn::canonical} and~\eqref{eqn::CPsuper}. That is, it is the operator sum (or  Kraus) representation for the process tensor, which (like the superchannel) implies that it is CP~\cite{pollock_complete_2015}. It also clearly satisfies the containment property, \textit{i.e.}, $\Tcal^{j:k} \subset \Tcal^{i:l} \, \forall \, i \le j \le k \le l$ and it is trace preserving. Indeed, the process tensor reduces to the \emph{superchannel} for a single-step process and it is the natural extension of the formalism in Sec.~\ref{sec::QuantumMaps} to more general experimental situations.

Comparable approaches to general quantum stochastic processes were already developed by Lindblad~\cite{lindblad_non-markovian_1979} and Accardi \textit{et al.}~\cite{accardi_nonrelativistic_1976, accardi_quantum_1982}, but have not gained traction with the community of researchers working on open quantum systems. The process tensor framework straightforwardly leads to several important results, most notably an operationally well-defined quantum Markov condition, measures of non-Markovianity (which we will briefly expand on below) ~\cite{pollock_complete_2015}, and a generalisation of the Kolmogorov extension theorem to general quantum stochastic processes~\cite{accardi_quantum_1982}.

Finally, similar mathematical structures (maps whose inputs and outputs are quantum maps themselves) have also been developed in other contexts, and are referred to as quantum combs~\cite{chiribella_transforming_2008, chiribella_quantum_2008, chiribella_theoretical_2009}, causal automata/non-anticipatory quantum channels~\cite{kretschmann_quantum_2005, caruso_quantum_2014}, process matrices for causal modelling~\cite{1367-2630-18-6-063032,oreshkov_causal_2016}, causal
boxes~\cite{portmann_causal_2015} and operator tensors~\cite{hardy_operational_2016,hardy_operator_2012}. Most of the results for the process tensor, including the representations we will now go on to describe, will also be applicable (or, at least, adaptable) to many of these frameworks. 

\subsection{Structure and representation of the superchannel and process tensor}

Since the process tensor, and hence the superchannel, are CP maps of the sort described in Sec.~\ref{subsec::LinMaps} with input and output spaces of different size,  we are able to represent them mathematically in all the ways discussed in the first half of this paper. For the most part, these representations are the same as for the quantum maps case (with the same mathematical properties), however it is insightful to present them explicitly. Given that the superchannel is simply a single step process tensor, its representations are a special case of what we will now present for an arbitrary number of time steps.

Performing quantum process tomography of process tensors is very similar to the usual case. At each time step $j$ we choose a basis set of linearly independent operations $\{\hat{\Acal}_{i_j} \}_{{i_j}=1}^{d^4}$. The index $i_j$ denotes both the basis element, as well as the time step, \textit{i.e.},  $\hat{\Acal}_{i_j}$ is the $i$th basis element at time step $j$. For example, at time step 3, we would have $\{\Acal_{1_3}, \ \Acal_{2_3}, \dots, \Acal_{d^4_3} \}$.  The basis elements at different times need not be the same, $\{\hat{\Acal}_{i_j}\} \ne \{\hat{\Acal}_{i_{j'}}\}$. An arbitrary control operation $\Acal_j$ at time step $j$ can be expressed as a linear combination of the basis operations $\Acal_j = \sum_{i_j} c_{i_j} \hat{\Acal}_{i_j}$. The basis operations come with a dual set $\{\hat{\Delta}_{i_j}\}$ satisfying $\tr[\hat{\Acal}_{i_j} \hat{\Delta}_{i'_j}^\dag] = \delta_{i_j i_j'}$. From the local basis, we can construct a basis sequence as $\hat{\mathbf{A}}_{\mathbf{i}_{k-1:0}} = \hat{\Acal}_{i_{k-1}} \otimes \dots \otimes \hat{\Acal}_{i_{0}}$, where $\mathbf{i}_{k-1:0} = (i_{k-1} \dots i_0)$. Naturally, we have $\tr[\hat{\mathbf{A}}_{\mathbf{i}'_{k-1:0}} \hat{\mathbf{D}}_{\mathbf{i}_{k-1:0}}^\dag] = \delta_{\mathbf{i}'_{k-1:0}\ \mathbf{i}_{k-1:0}}$, where $\hat{\mathbf{D}}_{\mathbf{i}_{k-1:0}} = \hat{\Delta}_{i_{k-1}} \otimes \dots \otimes \hat{\Delta}_{i_0}$. As before, using the basis operations we can express any (possibly correlated) sequence of control operation as $\mathbf{A}_{k-1:0} = \sum_{\mathbf{i}_{k-1:0}} \ c_{\mathbf{i}_{k-1:0}} \ \hat{\mathbf{A}}_{\mathbf{i}_{k-1:0}}$.

Tomography then involves performing a set of experiments where we apply each basis sequence of control operations $\hat{\mathbf{A}}_{\mathbf{i}_{k-1:0}}$ and measure the corresponding state $\hat{\rho}'_{k|\mathbf{i}_{k-1:0}}$.  In analogy with Eq.~\eqref{eqn::ActionMap} we can then write the action of the process tensor as
\begin{gather}\label{eqn::process-tomo}
\Tcal^{k:0}[\mathbf{A}_{k-1:0}] = \sum_k \hat{\rho}'_{k|\mathbf{i}_{k-1:0}} \ \tr \left[ \hat{\mathbf{D}}_{\mathbf{i}_{k-1:0}}^\dag \mathbf{A}_{k-1:0} \right],
\end{gather}
The set $\{\hat{\mathbf{D}}_{\mathbf{i}_{k-1:0}},\hat{\rho}'_{k|\mathbf{i}_{k-1:0}}\}$ constitutes the \textit{tomographic representation} of the process tensor.

From this, we can use the results of Sec.~\ref{subsec::Forms} to write down the process tensor in Sudarshan's A and B forms, in analogy to Eqs.~\eqref{eqn::EcalA} and~\eqref{eqn::EcalB}, as
\begin{gather}
\Tcal^{k:0}_A = \sum_k \hat{\rho}'_{k|\mathbf{i}_{k-1:0}} \times \hat{\mathbf{D}}_{\mathbf{i}_{k-1:0}}^* 
\qquad \mbox{and} \qquad
\Tcal^{k:0}_B = \sum_k \hat{\rho}'_{k|\mathbf{i}_{k-1:0}} \otimes \hat{\mathbf{D}}_{\mathbf{i}_{k-1:0}}^*.
\end{gather}
Similarly from the operator sum representation, given in Eq.~\eqref{eqn::processdilated}, we can write these in terms of Kraus operators (in analogy to Eqs.~\eqref{eqn::Kraus_A} and~\eqref{eqn::Kraus_B}):
\begin{align}
\Tcal^{k:0}_A = \sum_\alpha T^{k:0}_\alpha \otimes {T^{k:0}_\alpha}^* 
\qquad \mbox{and} \qquad
\Tcal^{k:0}_B = \sum_\alpha T^{k:0}_\alpha \times {T^{k:0}_\alpha}^*.
\end{align}
Just as in the case of quantum maps $\Ecal$, $\Tcal^{k:0}_A$ is a (rectangular) matrix which acts on a vectorised input -- in this case, it is the B form of the control operations $\mathbf{A}_{k-1:0}$ that is vectorised. In contrast, $\Tcal^{k:0}_B$ is a square matrix, whose positivity depends on the complete positivity of the process tensor. It acts as
\begin{align}
\label{eqn::ChoiSuper}
\Tcal^{k:0}[\mathbf{A}_{k-1:0}]=\tr_{\rm in} \left[\Tcal^{k:0}_B \left(\openone_{\rm out} \otimes \mathbf{A}_{k-1:0}^\mathrm{T} \right)\right]\, .
\end{align}
In fact, analogously to the case of quantum maps, the B form can be seen as arising from a generalisation of the CJI to process tensors~\cite{pollock_complete_2015}. The isomorphism can be implemented operationally by preparing $k$ maximally entangled states $\ket{I}$ (introduced at the beginning of Sec.~\ref{sec::CJI}) and swapping the system with one part of the maximally entangled state at each time step. Defining $\Psi_{\an b}$ as the superchannel describing an initial state $\ket{I}_{\an b}$ that later evolves under an identity map (\textit{i.e.}, with B form $\mathcal{I}_B\otimes \ketbra{I}{I}$), we can write down the CJI:
\begin{align}\label{eqn::ChoiState}
\Upsilon_\Tcal^{k:0} = \Tcal^{k:0}_\s \otimes \Psi^{\otimes k}_{\an b} [\mathcal{S}_{\s\an}^{\otimes k} \otimes \Ical_b^{\otimes k}],
\end{align}
where $\mathcal{S}_{\s\an}$ is the swap gate on $\s\an$ and $\Ical_b$ is the identity map on $b$ (see Fig.~\ref{fig::CJI}). The resultant $2k+1$ body state $\Upsilon_\Tcal^{k:0}$ is element by element identical to the process tensor. Again, using the equivalence of the B form and the Choi state, the action of the map can be written as $\Tcal^{k:0}[\mathbf{A}_{k-1:0}]=\tr_{\rm in}[\Upsilon^{k:0}_\Tcal \ (\openone_{\rm out} \otimes \mathbf{A}_{k-1:0}^\mathrm{T})]$.

\begin{figure}
\centering
\includegraphics[scale=0.6]{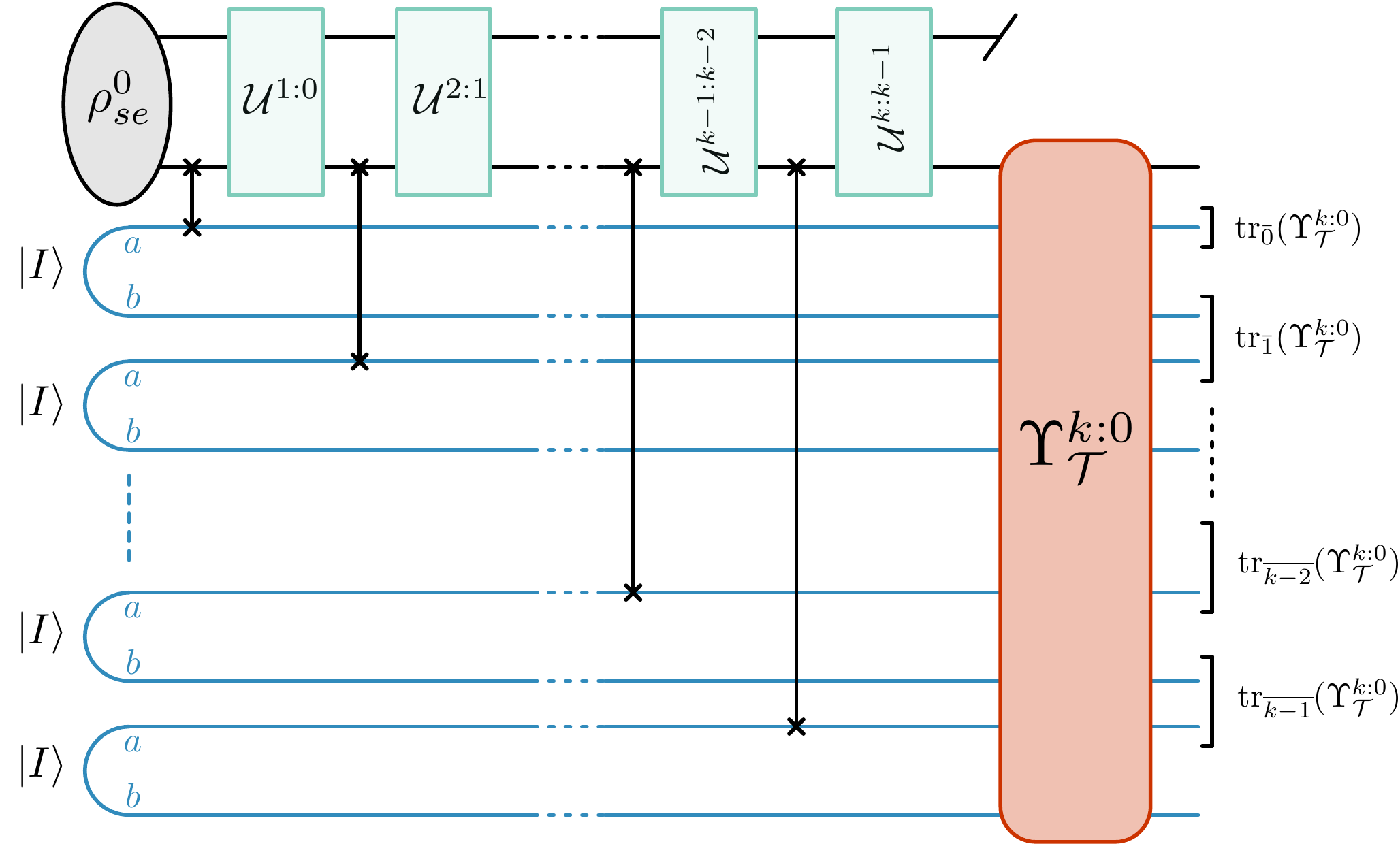}
\caption{\textit{CJI of a process tensor.} At each time $t_i$, one half (`the $a$ part') of a maximally entangled state is fed into the process by means of a swap operation (depicted by the black vertical lines). The resulting many-body quantum state (after tracing over the degrees of freedom of the environment) is the Choi state $\Upsilon^{k:0}_\mathcal{T}$ of $\mathcal{T}^{k:0}$. This generalised CJI includes the traditional CJI for quantum maps as well as the CJI of superchannels as special cases. The brackets denote the degrees of freedom of $\Upsilon^{k:0}_\mathcal{T}$ that correspond to the partial traces mentioned below Eq.~\eqref{eqn::CorrelationExpansion}. The remaining degrees of freedom -- \textit{i.e.}, the top and bottom wire -- correspond to $\tr_{\bar{k}}(\Upsilon^{k:0}_\Tcal)$.}
\label{fig::CJI}
\end{figure}

Since pairs of subsystems of the Choi state $\Upsilon^{k:0}_\Tcal$ correspond to different time steps, correlations between them directly relate to memory effects in the process. That is, temporal correlations in $\Tcal^{k:0}$ become spatial correlations in $\Upsilon^{k:0}_\Tcal$. This can most clearly be seen by decomposing the B form/Choi state as
\begin{align}
\label{eqn::CorrelationExpansion}
\Upsilon^{k:0}_\Tcal =& \Ecal_B^{k:k-1}\otimes\Ecal_B^{k-1:k-2}\otimes\cdots\otimes \Ecal_B^{1:0}\otimes \rho^0 \nonumber\\&
+ \sum_{j>j'} \chi_{j,j'} + \sum_{j>j'>j''} \chi_{j,j',j''} + \cdots + \chi_{k,k-1,\dots,1,0}, 
\end{align}
with $\tr_{j_i}[\chi_{j_n,\dots,j_1}] =0 \ \forall \ 1\leq i\leq n$, and where the indices in the sums represent pairs of subsystems belonging to the input of one time step and the output of the previous one (with the exception of $j=0$, which refers a single subsystem -- the initial input). We denote by $\tr_{\bar j}$ the partial trace over all subsystems but the ones that correspond to the dynamics from time $j-1$ to $j$ (see Fig.~\ref{fig::CJI}). With this, we obtain $\Ecal_B^{j:j-1} = \tr_{\bar{j}}[\Upsilon^{k:0}_\Tcal]$, the B form of a quantum map connecting an adjacent pair of time steps, and $\rho^0 = \tr_{\bar{0}}[\Upsilon^{k:0}_\Tcal]$ -- the initial, pre-preparation state of the system undergoing the process. The traceless matrices $\chi$ encode correlations between time steps, and it is precisely these which will contract with the B forms of measurement operations at different time steps in Eq.~\eqref{eqn::ChoiSuper} to produce multi-time correlation functions.

\subsection{Quantum Markov processes and measuring non-Markovianity}

If a classical stochastic process has no correlations between observables at different times, beyond those mediated by adjacent time steps, then it is called Markovian; formally
\begin{gather}\label{eqn::Cl-Markov}
P(X_k,t_k|X_{k-1},t_{k-1},\dots,X_0,t_0) = P(X_k,t_k|X_{k-1},t_{k-1}) \quad \forall k.
\end{gather}
Generalising Eq.~\eqref{eqn::Cl-Markov}, to give a necessary and sufficient condition for a quantum process to be Markovian has been a difficult task. In recent years, researchers have built a zoo of ``measures" of non-Markovianity~\cite{NMrev, RevModPhys.88.021002}. Most of these measures are based only on necessary conditions for classical processes to be Markovian\footnote{Some of these measures are claimed to be \textit{necessary and sufficient}, but only with respect to a quantum Markov condition which does not reduce to the classical one in the correct limit.}. For instance, the trace distance between two probability distributions must monotonically decrease under a classical Markov process. This should also be true for any pair of density matrices undergoing a quantum Markov process. Conversely, if the trace distance between any two density matrices does not decrease monotonically, then it implies that the underlying process is non-Markovian. A measure of non-Markovianity can be defined by summing up the non-monotonicity in time~\cite{PhysRevLett.103.210401}. Other witnesses are based on: how quantum Fisher information changes~\cite{PhysRevA.82.042103}; the detection of initial correlations~\cite{mazzola2012dynamical, rodriguez2012unification}; changes to quantum correlations~\cite{PhysRevA.86.044101}; positivity of quantum maps~\cite{PhysRevLett.101.150402, sabri}; and, most notably, witnessing the breakdown of the divisibility of a process~\cite{PhysRevLett.105.050403}. These witnesses are turned into measures by quantifying the degree to which they witness the departure from Markov dynamics.

All of these methods are perfectly valid ways of witnessing memory effects. Unfortunately though, they often lack a clear operational basis. Moreover, different measures of non-Markovianity do not always agree with each other, neither on the degree of non-Markovianity, nor on deciding whether a given process is Markovian~\cite{PhysRevA.83.052128}. In other words, each of them fails to quantify some demonstrable memory effects. These inconsistencies have led some researchers to conclude that there can be no unique condition for a quantum Markov process.

This is not correct. Using the process tensor framework it is possible to write down a \textit{necessary and sufficient} condition for quantum Markov processes~\cite{pollock_complete_2015},
that is mathematically unique and operationally sound. It encompasses all the other definitions
by objectively identifying all possible temporal correlations responsible for all possible memory effects -- including the correlations missed by the methods listed above (see examples in Ref.~\cite{pollock_complete_2015}). For classical dynamics, this quantum condition reduces to the Markov condition given in Eq.~\eqref{eqn::Cl-Markov}.

We can use the expansion in Eq.~\eqref{eqn::CorrelationExpansion} to make this condition more explicit. A process whose Choi state can be written as $\Upsilon^{k:0}_\Tcal = \Ecal_B^{k:k-1}\otimes\cdots\otimes \Ecal_B^{1:0}\otimes \rho^0$ (with all $\chi$ operators zero) will only lead to joint probability distributions which satisfy the Markov property for \emph{any} choice of measurements (not necessarily projective) at different time steps\footnote{Though the distributions for different choices of measurements will not be compatible in general (this could be seen as the defining feature of quantum theory).}, and we could take the product form as a definition for a quantum Markov process. Operationally, this means that a causal break in the system's evolution at any point prevents information flowing from past to future, see Ref.~\cite{pollock_complete_2015} for a rigorous derivation.

From this Markov condition we can derive a family of measures for non-Markovianity that are operationally meaningful for specific tasks. For instance, through Eq.~\eqref{eqn::CorrelationExpansion}, any process can be related to a corresponding Markov process $\Tcal^{k:0}_{\rm Mkv}$ (with Choi state $\Upsilon^{k:0}_{\Tcal_{\rm Mkv}} = \Ecal_B^{k:k-1}\otimes\cdots\otimes \Ecal_B^{1:0}\otimes \rho^0$) by simply setting the correlation terms to zero (removing the $\chi$'s). The total `amount' of memory in the process, or the \emph{degree of non-Markovianity} $\mathcal{N}$ can then be quantified by the distance of its Choi state from that of its Markov counterpart:
\begin{gather}\label{eqn::NMmeasure}
\mathcal{N} = \mathcal{D}\left(\Upsilon^{k:0},\Upsilon^{k:0}_{\Tcal_{\rm Mkv}}\right)
\end{gather}
where $\mathcal{D}$ could be any (pseudo) distance measure on quantum states, such as the trace distance. In particular, when relative entropy is chosen as the distance measure in Eq.~\eqref{eqn::NMmeasure}, the 
measure has a clear interpretation in terms of the probability of surprise $P_{\textrm{surprise}} = e^{-n \mathcal{N}}$. That is, suppose we have an experimental process that is non-Markovian and a model for this experiment that is Markovian. Then, after $n$ experiments how surprising are the results, given our Markov assumption? If $\mathcal{N}$ is small, then it will take many experiments (large $n$), before we observe statistically significant deviations in our data from the assumed model, and if $\mathcal{N}$ is large then we are surprised after only a small number of experiments $n$.

Different choices of distance will lead to different operational meanings for $\mathcal{N}$. Other measures of non-Markovianity could, for instance, indicate how much of the original state of $\s$ can be recovered, or how many extra degrees of freedom are needed to model the dynamics of $\s$ to a desired accuracy.

\section{Discussion and conclusions}

In this article, we have laid bare the operational motivation and underlying structure of quantum maps. We began by describing the familiar quantum maps that act on density operators and transform them into density operators, before going on to derive and relate their most widely used forms, and discuss their most important properties. While we worked only with finite dimensional systems, all of the maps presented here can also be extended to the case of infinite dimensional systems~\cite{davies76a, Brukner2016}.

Next, we described the problem of characterising quantum dynamics in the presence of initial correlations between a system and its environment. We outlined the attempts to describe such dynamics with not completely positive maps and the operational shortcomings of this approach. Then we presented a resolution to this problem in terms of the quantum superchannel, which generalises the quantum maps from the first section of the paper, and has all of the same desirable properties, like complete positivity and trace preservation. The development of the superchannel paved the way for us to introduce the process tensor framework, which can be used to describe any quantum process -- importantly, including its multi-time 
correlations. Major results enabled by this framework are a necessary and sufficient condition for quantum Markov processes and, consequently, a family of operationally meaningful measures for non-Markovianity.

The different mathematical representations we have presented arise from the statistical and linear algebraic framework on which quantum theory is based. In fact, they could also be used to describe a more general 
linear theory, such as one based on quaternionic vector spaces~\cite{Graydon2013}, as well as other generalised statistical theories~\cite{Barrett2007,MasanesMuller2011}. It is worth mentioning tensor network
calculus as a helpful tool for graphically representing quantum maps (and other linear algebraic objects). Diagrammatic proofs of the statements in this paper more clearly reveal the connections between different 
representations, as well as the similarity between the approaches of the first and second Section of this paper. For a comprehensive introduction in the context of open quantum systems theory, see Ref.~\cite{wood_tensor_2011}.

The process tensor is a powerful tool, and we have only just scratched the surface when it comes to unsolved open quantum dynamics problems. There remains a great deal of work to be done in order to better 
understand the properties of non-Markovian quantum processes. This includes, but is not limited to, characterising the length and strength of memory and investigating typical properties of random multi-time processes.
It remains to be seen whether something like the process tensor can be derived for setups where continuous control is applied, or where the experimenter's operations also influence the environment to some degree.

It should also be possible to use the process tensor framework to develop new methods for simulating open quantum systems. An approach based on tomographically reconstructed quantum maps has already been shown to be 
efficient~\cite{CerrilloCao2014, Rosenbach2016, PollockModi2017}, and it seems a natural step forward to generalise this to the multi-time case. Furthermore, such a method would be easy to adapt into a simpler 
approximate description; the process tensor quantifies exactly the observable influence of the environment on a system, therefore its structure should indicate exactly which quantities can be safely neglected in the 
global dynamics.

\begin{acknowledgements}
KM thanks Usha Devi, V. Jagdish, and A. K. Rajagopal for discussions. KM is supported through ARC FT160100073. SM is supported by the Monash Graduate Scholarship (MGS), Monash International Postgraduate Research Scholarship (MIPRS) and the J L William Scholarship.
\end{acknowledgements}

\appendix
\section*{APPENDIX}
\section{Dual matrices}
\label{subsec::DualMat}

In this appendix, we prove the existence of a set of dual matrices for any set of linearly independent matrices $\rho_i$. Note, that this proof is a slight generalisation of the one presented in~\cite{modi_positivity_2012} for the case of Hermitian matrices $\rho_i$.

\begin{lemma} For any set of Linearly independent matrices $\{\rho_i\}$, there exists the dual set $\{\hat{D}_i\}$ satisfying $\mathrm{tr}[\hat{D}_i^{\dagger} \; \rho_j] =\delta_{ij}$.
\end{lemma}

\begin{proof} Write $\rho_i = \sum_j h_{ij} \Gamma_j$, where $h_{ij}$ are complex numbers and $\{\Gamma_j\}$ form a Hermitian self-dual linearly independent basis satisfying $\tr[\Gamma_i \Gamma_j]=2 \delta_{ij}$~\cite{byrd:062322}. Since $\{\rho_i\}$ constitute a linearly independent set, the columns of matrix $\mathsf{H} = \sum_{ij} h_{ij} \ketbra{i}{j}$ are linearly independent vectors, which means that $\mathsf{H}$ has an inverse. Let the matrix $\mathsf{F}^\dagger=\mathsf{H}^{-1}$, then $\mathsf{H} \mathsf{F}^\dagger = \openone$, implying that the columns of $\mathsf{F}^*$ are orthonormal to the columns of $\mathsf{H}$. We define $\hat{D}_i = \frac{1}{2} \sum_j f_{ij} \Gamma_j$, where $f_{ij}$ are elements of $\mathsf{F}$.
\end{proof}
Our definition of dual matrices differs from the one in~\cite{modi_positivity_2012} by an adjoint to make the relation to the scalar product explicit. As already mentioned, the dual matrices are generally not all positive, even if the basis $\{\rho_i\}$ only consists of positive matrices. However, for the case where all basis matrices $\rho_i$ are Hermitian, we have $\hat{D}_i^{\dagger} = \hat{D}_i$. Furthermore, the duals satisfy $\sum_i \hat{D}_i^\dag = \sum_i \hat{D}_i^* = \openone$ if all $\rho_i$ are of unit trace. We have
\begin{gather}
\tr \left(\sum_i\hat{D}_i^{\dagger}\rho \right) = \sum_{i,j} r_j \tr\left(\hat{D}_i^\dagger \rho_j\right) = \sum_j r_j = \tr(\rho) \quad \forall \rho\, ,
\end{gather}
where we have used $\rho = \sum_j r_j \rho_j$. The only matrix $M$ that satisfies $\tr(M\rho) = \tr(\rho) \ \forall \rho$ is the identity matrix.

\FloatBarrier

\bibliographystyle{apsrev4-1}
\bibliography{product}

\end{document}